\documentclass[a4paper,10pt]{article}

\usepackage{amsthm}
\usepackage{amsmath}
\usepackage{amssymb}
\usepackage{amsfonts}
\usepackage{amscd}
\usepackage{bm}
\usepackage[dvips]{graphicx}

\newtheorem{thm}{Theorem}[section]
\newtheorem{rem}{Remark}[section]

\newtheorem{prp}{Proposition}[section]
\newtheorem{lem}{Lemma}[section]

\makeatletter
\@addtoreset{equation}{section}
\makeatother

\def\;{{\hspace{0.3ex};\hspace{0.5ex}}}
\def\,{{\hspace{0,3ex},\hspace{0.5ex}}}
\def\({{\hspace{1.2ex}(}}

\def\R{{\mathbb{R}}}
\def\Z{{\mathbb{Z}}}
\def\N{{\mathbb{N}}}

\def\C{{\mathbb{C}}}
\def\ph{{\varphi}}

\def\im{\mathop{\bf Im}\nolimits}
\def\re{\mathop{\bf Re}\nolimits}

\def\Tr{\mathop{\rm Tr}\nolimits}

\def\im{\mathop{\bf Im}\nolimits}

\def\R{\mathbb{R}}
\def\C{\mathbb{C}}

\def\im{\mathop{\rm Im}}
\def\re{\mathop{\rm Re}}

\def\QED{\mbox{\rule[0pt]{1.5ex}{1.5ex}}}

\def\endproof{\hspace*{\fill}~\QED\par\endtrivlist\unskip}

\def\Label#1{\label{#1}\ [\ #1\ ]\ }
\def\Label{\label}

\title{Quantum hypothesis testing for quantum Gaussian states: Quantum analogues of chi-square, t and F tests}
\author{Wataru Kumagai\footnote{kumagai@ims.is.tohoku.ac.jp}~~~~~Masahito Hayashi\footnote{hayashi@math.is.tohoku.ac.jp}}

\date{$^{*,\dagger}$Graduate School of Information Sciences, Tohoku University, Japan \\$^{\dagger}$Centre for Quantum Technologies, National University of Singapore, Singapore}
\begin{document}

\maketitle

\begin{abstract}
We treat quantum counterparts of testing problems whose optimal tests are given by $\chi^2$, $t$ and $F$ tests.
These quantum counterparts are formulated as quantum hypothesis testing problems concerning quantum Gaussian states families,
and contain disturbance parameters, which have group symmetry.
Quantum Hunt-Stein Theorem removes a part of these disturbance parameters, but 
other types of difficulty still remain.
In order to remove them,
combining quantum Hunt-Stein theorem and other reduction methods, 
we establish a general reduction theorem that reduces a complicated quantum hypothesis testing problem to a fundamental quantum hypothesis testing problem.
Using these methods, 
we derive quantum counterparts of $\chi^2$, $t$ and $F$ tests as optimal tests in the respective settings.
\end{abstract}

%%%%%%%%%%%%%%%%%%%%%%%%%%%%%%%%%%%%%%%%%%%%%%%%%%%%%%%%%%%%%%%%
%%%%%%%%%%%%%%%%%%%%%%%%%%%%%%%%%%%%%%%%%%%%%%%%%%%%%%%%%%%%%%%%
\section{Introduction}

In recent years, movement for the achievement of the quantum information processing technology has been activated. 
Since it is necessary to prepare a quantum state and to manipulate it accurately
for the quantum information processing, 
we need a proper method to decide whether the realized quantum state is the intended state.
For such a situation, we have to prepare two hypotheses.
One is the null hypothesis $H_0$, which corresponds to the undesired case.
The other hypothesis is called the alternative hypothesis $H_1$, which corresponds to the desired.
Then, both hypotheses $H_0$ and $H_1$ are incompatible with each other.
In order to guarantee the desired property, it is sufficient to show that
the null hypothesis $H_0$ is not true.
This type formulation is called quantum hypothesis testing,
which is the quantum counterpart of statistical hypothesis testing.

In this setting, when the hypothesis consists of one element,
it is called a simple hypothesis and can be easily treated even in the quantum case.
In particular, when both hypotheses $H_0$ and $H_1$ are simple, 
a construction method of an optimal test is guaranteed by the Neyman-Pearson lemma (\cite{Her}, \cite{Hol2}), 
and the asymptotic performance such as Stein's lemma (\cite{N-O},\cite{Hay6}),  Chernoff bound (\cite{N-S}, \cite{Aud.et al.}, \cite{Aud.et al.2}) and Hoeffding bound (\cite{Nag}, \cite{Hay4}, \cite{Aud.et al.2}) are studied. 
Such a case has been studied extensively, including the quantum case.
On the other hand, when the hypothesis consists of plural elements,
it is called a composite hypothesis.
In realistic cases, it is usual that either $H_0$ or $H_1$ is composite.
In the classical hypothesis testing, composite hypothesis cases are well studied.
Especially, testing problems for the Gaussian distribution frequently appear with various composite hypotheses. 
When there is disturbance in the basic hypothesis testing problems for the Gaussian distributions, optimal tests are given by $\chi^2$, $t$ and $F$ tests if the unbiasedness is offered for tests (\cite{Leh2}). 
These three kinds of testings are mostly applied in the classical hypothesis testing.
However, compared to classical hypothesis testing, the composite hypothesis case is not well studied in quantum hypothesis testing.
There exist only a few studies for testing entanglement among quantum hypothesis testing with a composite hypothesis (\cite{Hay1}, \cite{HMT}). 
In particular, no quantum counterparts of $\chi^2$, $t$ and $F$ tests have been studied.

In this paper, 
we focus on the quantum Gaussian state (\cite{Hol1}) that is a Gaussian mixture of quantum coherent states in a bosonic system.
Quantum Gaussian states $\{\rho_{\theta,N}\}_{\theta\in\C,N\in\R_{>0}}$ have two parameters called the mean parameter $\theta\in\C$ and the number parameter $N\in\R_{>0}$ corresponding to the mean parameter and the variance parameter of Gaussian distributions. 
This kind of states frequently appear in quantum optical systems.
For testing these states, we propose quantum counterparts of $\chi^2$, $t$ and $F$ tests for quantum Gaussian states.
The importance of Gaussian distributions relies upon the (classical) central limit theorem. 
That is, the central limit theorem guarantees that
$\chi^2$, $t$ and $F$ tests can be applied to the i.i.d.(independent and identical distributed) case
even though the distribution family is not the Gaussian distribution family.
Similarly, when many quantum states are independently and identically given, 
the quantum states are approximated by a quantum Gaussian state 
due to the quantum central limit theorem \cite{Petz1,Guta1,Guta2,Guta3,Hay2,Hay7,B-G}. 
Therefore, it can be expected that quantum Gaussian sates can play the same role of Gaussian distributions.
Hence, our quantum counterparts can be applied to many realistic quantum cases.

The main difficulty is treatment of disturbance parameters
because our quantum hypothesis testing problems contain them. 
For their proper treatment, we sometimes adopt the min-max criteria for the disturbance parameters.
That is, we optimize the worst value of the error probability concerning the disturbance parameters.
Since our problem is hypothesis testing,
our target is deriving Uniformly Most Powerful (UMP) tests in the respective problems.
Under this framework, 
we focus on unitary group representation on the disturbance parameters.
Although the group symmetry method has been much succeeded in state estimation\cite{Hay7,Bagan,Hol,Hol1,M-P,Aspachs,Hay8},
it can be applied to quantum hypothesis testing in a few cases\cite{Hay1,HMT,Hiai}.
Under the respective symmetries, the UMP min-max test can be obtained among invariant tests.
This argument in a more general setting is justified by quantum Hunt-Stein theorem\cite{Hol}.

%As for the details, $\chi^2$ distribution is composed by Gaussian random variables and invariant with respect to the parallel translation to the Gaussian random variables. 
%Similarly, $t$ and $F$ distributions are composed by Gaussian random variables too and invariant with respect to constant multiple to the Gaussian random variables. 

However, there exists other type of difficulty except for disturbance parameters.
Hence, combining quantum Hunt-Stein theorem with other reduction methods, 
we establish a general reduction theorem that reduces a complicated quantum hypothesis testing problem to a fundamental quantum hypothesis testing problem.
This general theorem enables us to translate 
our quantum hypothesis testing problems to 
fundamental testing problems related to $\chi^2$, $t$ and $F$ tests. 
Another type of difficulty still remains even after the application of the above methods, 
but can be resolved by employing known facts in classical statistics.
Overall, we treat 8 quantum hypothesis testing problems 
for quantum Gaussian states families as given in Table 1.
Their treatment is summarized in Fig. \ref{f0}.

\begin{table}[htb]
 \begin{center}
  \begin{tabular}{|c||c|c||c|c|} \hline
  \multicolumn{1}{|c||}{Testing problem}
      & \multicolumn{2}{|c||}{Gaussian distribution } & \multicolumn{2}{|c|}{Quantum Gaussian state}\\ \hline \hline
  & $N:$ known & $N:$ unknown & $N:$ known & $N:$ unknown \\ \hline
  $|\theta|\le R_0~vs.~|\theta|> R_0$ & N-P lem. + MLR & $t$ test ($R_0=0$) & \textbf{5.1} & \textbf{5.2} ($R_0=0$) \\ \hline
  $\theta=\eta~vs.~\theta\ne\eta$  & N-P lem. + MLR & $t$ test & \textbf{6.1} & \textbf{6.2} \\ 
\hline \hline
  & $\theta:$ known & $\theta:$ unknown & $\theta:$ known & $\theta:$ unknown \\ \hline
  $N\le N_0~vs.~N> N_0$  & $\chi^2$ test & $\chi^2$ test & \textbf{7.1} & \textbf{7.2} \\ \hline
  $M=N~vs.~M\ne N$   & $F$ test & $F$ test & \textbf{8.1} & \textbf{8.2} \\ \hline
  \end{tabular}
 \end{center}
\caption{This table shows well-known results on classical hypothesis testing for Gaussian distributions 
and summary in this paper on quantum hypothesis testing for quantum Gaussian states. 
N-P lem. and MLR means Neyman-Pearson's lemma and monotone likelihood ratio, respectively.}
\end{table}

\begin{figure}[htbp]
\begin{center}
\scalebox{1.0}{\includegraphics[scale=0.8]{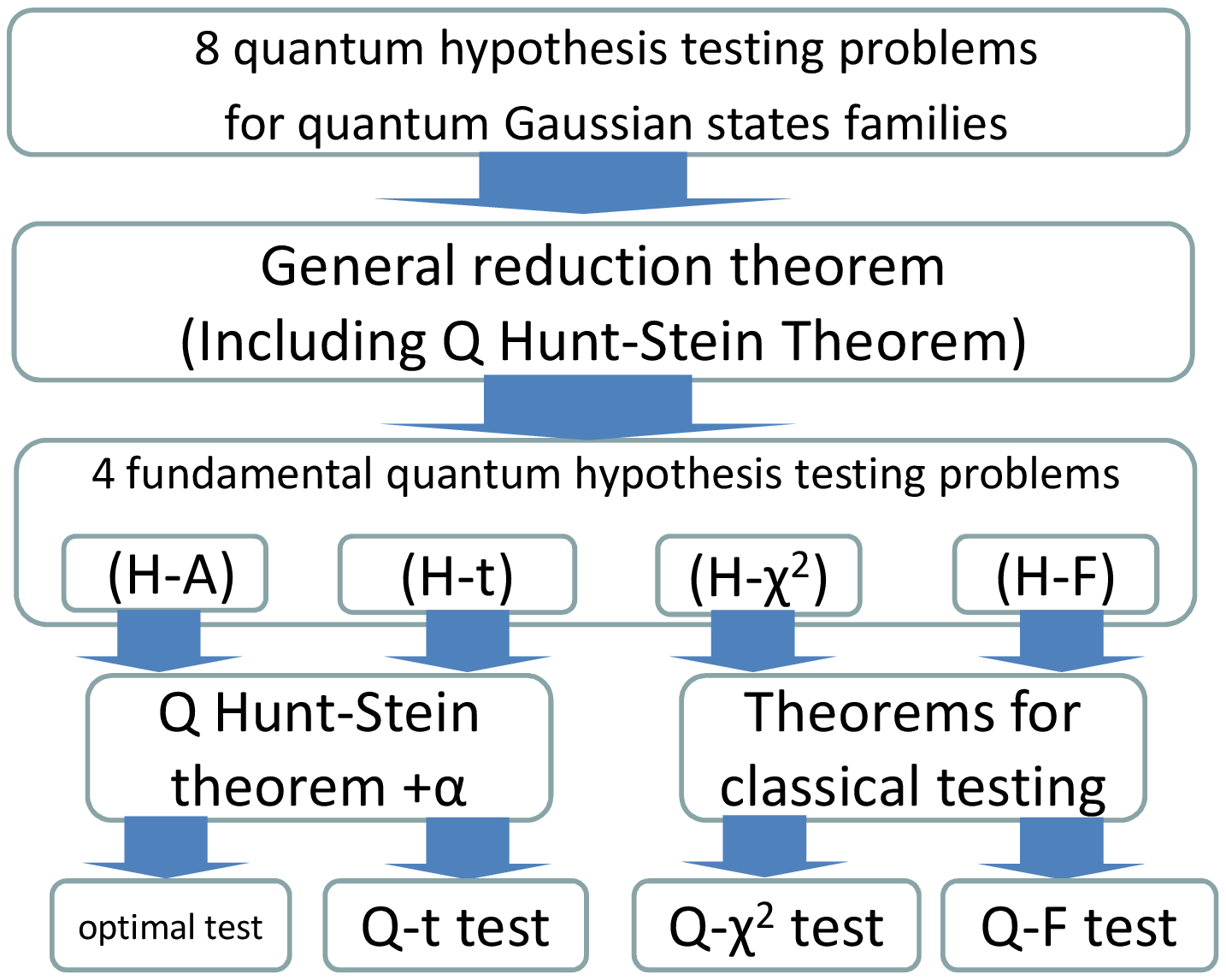}}
\end{center}
\caption{
Our strategy for 8 quantum hypothesis testing problems 
for quantum Gaussian states families.}
\Label{f0}
\end{figure}%

This paper is organized as follows. 
In section 2, we review fundamental knowledges concerning the bosonic system and the quantum Gaussian states.
In particular, several useful unitary operators are explained.
In section 3, the framework of quantum hypothesis testing are stated and we restate quantum Hunt-Stein theorem 
in a min-max sense with the non-compact case. Thereafter, we prepare some reduction methods on quantum hypothesis testing.
In section 4, basic facts in classical hypothesis testing are reviewed.
Using them, we derive optimal test for fundamental testing problems related to $\chi^2$ and $F$ tests. 

In sections 5 to 8, using quantum Hunt-Stein theorem and our reduction theorem, 
we derive the optimal tests concerning quantum Gaussian states $\rho_{\theta,N}$
in the respective settings, as is summarized in Table 1.
Especially, quantum counterpart of $t$ test is derived in section 5.
These optimal tests are constructed by quantum versions of $\chi^2$, $t$ and $F$ tests,
which are based on the number measurement.
In section 9, we numerically compare the performance for our optimal tests 
and tests based on the combination of the classical optimal test and 
the heterodyne measurement, which is the optimal measurement in the sense of state estimation.
This comparison clarifies the advantage of our optimal tests. 
In section 10, we treat the relation of our quantum $\chi^2$ test with quantum state estimation.
In section 11, we give some concluding remarks.

%%%%%%%%%%%%%%%%%%%%%%%%%%%%%%%%%%%%%%%%%%%%%%%%%%%%%%%%%%%%%%%
%%%%%%%%%%%%%%%%%%%%%%%%%%%%%%%%%%%%%%%%%%%%%%%%%%%%%%%%%%%%%%%
\section{Bosonic system and quantum Gaussian state}

In this section, we introduce quantum Gaussian states in the bosonic system and some operations on the system. 
A single-mode bosonic system is mathematically represented by $\mathcal{H}=L^2(\R)$,
which is spanned by sets $\{|k\rangle\}_{k\in\Z_{\ge 0}}$ of the $k$-th Hermitian functions $|k\rangle$.
A most typical example of a single-mode bosonic system is the one-mode photonic system,
in which, 
the vector $|k\rangle$ is called the $k$-photon number state because it regarded as the state corresponding to $k$ photons.
Then, 
$M_{N}=\{|k\rangle\langle k|\}_{k\in\Z_{\ge 0}}$ forms a POVM, which is called the \textbf{number measurement}.  
When the vector $|\xi)\in L^2(\R)~(\xi\in\C)$ is defined by
\[
|\xi):=e^{-\frac{|\xi|^2}{2}}\sum_{k=0}^{\infty}\frac{\xi^k}{\sqrt{k!}}|k\rangle,
\]
the state $|\xi)(\xi|$ is called the coherent state because it corresponds to the coherent light in the quantum optical system. Then, the \textbf{quantum Gaussian state} 
is defined as a Gaussian mixture of coherent states in the following way:
\[
\rho_{\theta, N}=\frac{1}{\pi N}\int_{\C}|\xi)(\xi|e^{-\frac{|\theta-\xi|^2}{N}}d\xi,
\]
The mean parameter $\theta$ corresponds to the mean parameter of the Gaussian distribution,
and the number parameter $N$ does to the variance parameter of the Gaussian distribution.

In the bosonic system $L^2(\R)$, we focus on the momentum operator $Q$ and the position operator $P$, 
which are defined by 
\[
(Qf)(x):=xf(x),~~ (Pf)(x):=-i\frac{df}{dx}(x).
\]
Then, 
the \textbf{mean shift operator}
\begin{eqnarray}\Label{mean shift operator}
W_\theta:=\exp i(-\sqrt{2}\textbf{Re}\theta P + \sqrt{2}\textbf{Im}\theta Q)~~~(\theta\in\C)
\end{eqnarray}
plays the same role as 
the constant addition for a random variable distributed to a Gaussian distribution. 
That is, it satisfies
\[
W_{\theta'}\rho_{\theta, N}W_{\theta'}^{\ast}=\rho_{\theta+\theta', N}
\]
for any $\theta, \theta'\in\C$. We prepare a lemma about the mean shift operator.

\begin{lem}\Label{inv1}
For a bounded self-adjoint operator $T$ on $L^2(\R)^{\otimes n}$, $T$ satisfies $(W_{\xi}\otimes I^{\otimes (n-1)})T(W_{\xi}\otimes I^{\otimes (n-1)})^{\ast}=T$ for any $\xi\in\C$ if and only if $T$ is represented as $T=I\otimes T'$ where $T'$ is a bounded self-adjoint operator on $L^2(\R)^{\otimes (n-1)}$.
\end{lem}

Since the representation $\{W_{\theta}\}_{\theta\in\C}$ of $\C$ on $L^2(\R)$ is irreducible, the above lemma is followed by the Schur's lemma.

Using the number operator $\hat{N}:=\displaystyle\sum_{k=0}^{\infty}k|k\rangle\langle k|$,
we define the \textbf{phase shift operator}
\begin{eqnarray}\Label{phase shift operator}
S_{e^{it}}:=\mathrm{exp}(it\hat{N})~~~(e^{it}\in S^1:=\{a\in\C||a|=1\}),
\end{eqnarray}
which satisfies
\[
e^{it\hat{N}}\rho_{\theta, N}e^{-it\hat{N}}=\rho_{e^{it}\theta, N}
\]
for any $t\in\R, \theta\in\C$. 
The phase shift operation plays the same role as the rotation for a random variable distributed to a two-dimensional Gaussian distribution. We prepare a lemma about the phase shift operator.

\begin{lem}\Label{inv2}
For a bounded self-adjoint operator $T$ on $L^2(\R)^{\otimes n}$, $T$ satisfies
\begin{equation}\Label{S-invariant}
(S_{t}\otimes I^{\otimes (n-1)})T(S_{t}\otimes I^{\otimes (n-1)})^{\ast}=T
\end{equation}
for any $t\in\R$ if and only if $T$ is represented as
\begin{equation}\Label{diagonal}
T=\displaystyle\sum_{k=0}^{\infty}|k\rangle\langle k|\otimes T_k
\end{equation}
where each $T_k$ is a bounded self-adjoint operator on $L^2(\R)^{\otimes (n-1)}$.
\end{lem}

\begin{proof}
Let $T$ satisfy the equation (\ref{S-invariant}). When we denote $T$ by $\sum t_{k_1,\cdot\cdot\cdot,k_n, l_1,\cdot\cdot\cdot,l_n}|k_1,\cdot\cdot\cdot,k_n\rangle\langle l_1,\cdot\cdot\cdot,l_n|$, 
\[
e^{it_1(k_1-l_1)}t_{k_1,\cdot\cdot\cdot,k_n, l_1,\cdot\cdot\cdot,l_n}=t_{k_1,\cdot\cdot\cdot,k_n, l_1,\cdot\cdot\cdot,l_n}
\]
holds for $t_1\in\R,k_1,\cdot\cdot\cdot k_n,l_1,\cdot\cdot\cdot,l_n\in Z_{\ge 0}$ by the equation (\ref{S-invariant}). Hence we get $t_{k_1,\cdot\cdot\cdot,k_n, l_1,\cdot\cdot\cdot,l_n}=0$ for $k_1\ne l_1$. When we define as 
\[
T_k:=\displaystyle\sum_{k_2,\cdot\cdot\cdot,k_n,l_2,\cdot\cdot\cdot,l_n\in\Z_{\ge0}}t_{k,k_2,\cdot\cdot\cdot,k_n, k,l_2,\cdot\cdot\cdot,l_n}|k_2,\cdot\cdot\cdot,k_n\rangle\langle l_2,\cdot\cdot\cdot,l_n|,
\]
 $T$ is represented as (\ref{diagonal}), and each $T_k$ is clearly a bounded self-adjoint operator.

Conversely, when $T$ is represented as (\ref{diagonal}), then $T$ clearly satisfies (\ref{S-invariant}).
\end{proof}

In the $n$-mode bosonic system,
we denote 
the momentum operator and the position operator on the $j$-th bosonic system $\mathcal{H}_j=L^2(\R)~(j=1,\cdot\cdot\cdot, n)$ by $Q_j$ and $P_j$.
When the interaction Hamiltonian is given as $H_{j,j+1}:=i(a_{j+1}^{\ast}a_j-a_j^{\ast}a_{j+1})$
with $a_j:=\frac{1}{\sqrt{2}}(Q_j+iP_j)$, 
the coherent state $|\alpha_j)\otimes|\alpha_{j+1})$ on $\mathcal{H}_j\otimes\mathcal{H}_{j+1}$ is transformed to 
\[
\mathrm{exp}(itH_{j,j+1})|\alpha_j)\otimes|\alpha_{j+1})=|\alpha_j\mathrm{cos}t+\alpha_{j+1}\mathrm{sin}t)\otimes|-\alpha_j\mathrm{sin}t+\alpha_{j+1}\mathrm{cos}t).
\]
Hence, we get
\[
\mathrm{exp}(itH_{j,j+1})(\rho_{\theta_j,N}\otimes\rho_{\theta_{j+1},N})\mathrm{exp}(-itH_{j,j+1})=\rho_{\theta_j\mathrm{cos}t+\theta_{j+1}\mathrm{sin}t,N}\otimes\rho_{-\theta_j\mathrm{sin}t+\theta_{j+1}\mathrm{cos}t,N}
\]
for a quantum Gaussian state $\rho_{\theta_j,N}\otimes\rho_{\theta_{j+1},N}$. By choosing suitable interaction time $t_{1,2}, \cdot\cdot\cdot, t_{n-1,n}$, the unitary operator $U_n:=\exp(it_{1,2}H_{1,2})\cdot\cdot\cdot\exp(it_{n-1,n}H_{n-1,n})$ satisfies
\begin{eqnarray}\Label{concentrating}
U_n \rho_{\theta, N}^{\otimes n}U_n^{\ast}=\rho_{\sqrt{n}\theta, N}\otimes \rho_{0, N}^{\otimes(n-1)}
\end{eqnarray}
for any $\theta\in\C, N\in\R_{> 0}$. 
We call the above unitary operator $U_n$ satisfying the condition (\ref{concentrating}) the \textbf{concentrating operator}. In the same way, we can construct the unitary operators $U'_{m,n}$ and ${U''}_{m,n}$ satisfying
\begin{eqnarray}
U'_{m,n}(\rho_{\theta,N}^{\otimes m}\otimes\rho_{\eta,N}^{\otimes
n}){U'}_{m,n}^{\ast}
&=&\rho_{c_1(m\theta+n\eta),N}\otimes\rho_{c_0(\theta-\eta),N}\otimes\rho_{0,N}^{\otimes
(m+n-2)}\Label{U2},\\
U''_{m,n}(\rho_{\theta,N}^{\otimes m}\otimes\rho_{\eta,N}^{\otimes
n}){U''}_{m,n}^{\ast}
&=&
\rho_{\sqrt{m}\theta,M}\otimes\rho_{\sqrt{n}\eta,N}\otimes\rho^{\otimes(m-1)}_{0,M}\otimes\rho^{\otimes(n-1)}_{0,N}\Label{U3}
\end{eqnarray}
for any $\theta,\eta\in\C$ where $c_0:=\sqrt{m}\mathrm{cos}(\mathrm{tan}^{-1}(\sqrt{\frac{m}{n}}))$, $c_1:=\frac{1}{\sqrt{n}}\mathrm{cos}(\mathrm{tan}^{-1}(\sqrt{\frac{m}{n}}))$.

\begin{lem}\Label{prob}
When the number measurement $M_N$ is performed for the system with the quantum Gaussian state $\rho_{\theta, N}$,
the measured value $k$ is obtained with the probability
\[
P_{\theta, N}(k):=\langle k|\rho_{\theta, N}|k\rangle = \frac{1}{N+1}\left(\frac{N}{N+1}\right)^k e^{-\frac{|\theta|^2}{N+1}}L_k\left(-\frac{|\theta|^2}{N(N+1)}\right)
\]
where $L_k(x)=\displaystyle\sum_{j=0}^k\binom{k}{j}\frac{(-x)^j}{j!}$ is the $k$-th Laguerre polynomial.
\end{lem}

\begin{proof}
By using
\begin{align*}
& P_{\theta, N}(k)= \langle k|\rho_{\theta, N}|k\rangle
=\frac{1}{\pi k! N}\int_{\C}|\xi|^{2k} e^{-|\xi|^2} e^{-\frac{|\theta-\xi|^2}{N}}d\xi\\
=&\frac{e^{-\frac{|\theta|^2}{N+1}}} {\pi k! N} \int_{\C} |\xi|^{2k}
e^{-(1+\frac{1}{N})|\xi-\frac{\theta}{N+1}|^2}{d\xi} \\
=&\frac{1}{N+1}\left(\frac{N}{N+1}\right)^k e^{-\frac{|\theta|^2}{N+1}} \frac{1}{\pi k!} \int_{\C}
\Big|\xi+\frac{\theta}{\sqrt{N(N+1)}}\Big|^{2k} e^{-|\xi|^2}{d\xi}\\
\end{align*}
and
\begin{align}
&  \int_{\C}|\xi+c|^{2k}e^{-|\xi|^2}d\xi
=  \int_{\C}|\xi|^{2k}e^{-|\xi-|c||^2}d\xi\nonumber\\
=& \int_0^{\infty}\int_0^{2\pi}r^{2k} e^{-|re^{i\phi}-|c||^2}d\phi dr 
=  e^{-|c|^2}\int_0^{\infty}r^{2k+1}e^{-r^2}\int_0^{2\pi}e^{2|c|r cos\phi}d\phi dr\nonumber\\
=& e^{-|c|^2}\int_0^{\infty}r^{2k+1}e^{r^2}2\pi J_0(-2i|c|r) dr
=  \pi k! L_k(-|c|^2),
\Label{formula}
\end{align}
where $J_0(z)=\frac{1}{2\pi}\int_0^{2\pi}e^{iz cos\phi}d\phi$ is Bessel function of the first kind, the equation
\[
P_{\theta, N}(k)= \frac{1}{N+1}\left(\frac{N}{N+1}\right)^k e^{-\frac{|\theta|^2}{N+1}}L_k\left(-\frac{|\theta|^2}{N(N+1)}\right)
\]
holds. The fourth equation in (\ref{formula}) is derived by \cite[6.631,p738]{Table}.

\end{proof}

The probability $P_{\theta, N}$ in Lemma \ref{prob} can be calculated as follows in the specific situation.
When $\theta=0$, the distribution $P_{0, N}$ is the geometric distribution with the mean $N$,
i.e., 
\[
P_{0, N}(k)=\frac{1}{N+1}\left(\frac{N}{N+1}\right)^k
\]
holds for $k\in\Z_{\ge 0}$.
Similarly, when $N=0$, the distribution $P_{\theta, 0}$
is the Poisson distribution with the mean $|\theta|$, i.e., 
\[
P_{\theta, 0}(k)=\frac{|\theta|^k}{k!}e^{-|\theta|}
\]
holds for $k\in\Z_{\ge 0}$.

The mean and the variance of $P_{\theta, N}$ are represented as follows.

\begin{lem}\Label{moment}\cite{L}
The mean and the variance of the probability measure $P_{\theta, N}$ in Lemma \ref{prob} are $|\theta|^2+N$ and $N(N+1)+|\theta|^2 (2N+1)$ respectively.
\end{lem}

See \cite{L} for the higher moment of $P_{\theta, N}$. 

%%%%%%%%%%%%%%%%%%%%%%%%%%%%%%%%%%%%%%%%%%%%%%%%%%%%%%%%%%%%%%%%%%%%%
%%%%%%%%%%%%%%%%%%%%%%%%%%%%%%%%%%%%%%%%%%%%%%%%%%%%%%%%%%%%%%%%%%%%%

%%%%%%%%%%%%%%%%%%%%%%%%%%%%%%%%%%%%%%%%%%%%%%%%%%%%%%%%%%%%%%%%%%%%
%%%%%%%%%%%%%%%%%%%%%%%%%%%%%%%%%%%%%%%%%%%%%%%%%%%%%%%%%%%%%%%%%%%%
\section{Quantum hypothesis testing}

In this section, we describe the formulation for the quantum hypothesis testing.
In particular, we focus on the case when the composite hypotheses are given with the disturbance parameter.
In order to treat such a situation with symmetry,
we provide quantum Hunt-Stein theorem in the context of quantum hypothesis testing.
In this paper,
the dimension of the quantum system $\mathcal{H}$ of interest is not necessary finite,
and is assumed to be at most countable.

%%%%%%%%%%%%%%%%%%%%%%%%%%%%%%%%%%%%%%%%%%%%%%%%%%%%%%%%%%%%%%%%%%%%
%%%%%%%%%%%%%%%%%%%%%%%%%%%%%%%%%%%%%%%%%%%%%%%%%%%%%%%%%%%%%%%%%%%%
\subsection{Basic formulation for quantum hypothesis testing}\Label{HT}

~

In quantum hypothesis testing, 
in order to describe our hypotheses, the null hypothesis $H_0$ and the alternative hypothesis $H_1$,
we introduce two disjoint sets $\mathcal{S}_0$ and $\mathcal{S}_1$ of quantum states so that
the unknown state $\rho$ belongs to the union set $\mathcal{S}:=\mathcal{S}_0 \cup \mathcal{S}_1$.
Then, our problem and our two hypotheses $H_0$ and $H_1$ are described in the following way:
\[
H_0:\rho\in\mathcal{S}_0~~\text{vs.}~~H_1:\rho\in\mathcal{S}_1.
\]
What we should do in quantum hypothesis testing is to determine 
whether $\rho$ belongs to $\mathcal{S}_0$ or $\mathcal{S}_1$ 
by applying a two-valued POVM $\{T_0, T_1\}$ to the quantum system with the unknown state. 
In this method, we support the hypothesis $H_i$ when the outcome is $i$.
Since an arbitrary two-valued POVM $\{T_0, T_1\}$ is represented by an operator $0\le T\le I$ 
($I$ is the identity operator) as $T_0=I-T,T_1=T$, 
an operator $0\le T\le I$ is called a \textbf{test (operator)} in hypothesis testing. 
There are two kinds of erroneous in the above decision way: 
to accept $H_1$ although $H_0$ is true and 
to accept $H_0$ although $H_1$ is true, 
which are called type I error and type II error respectively. 
Type I error and type II error probabilities are expressed as 
\begin{equation}
\alpha_{T}(\rho):=\Tr \rho T~~(\rho\in\mathcal{S}_0), ~~~~~ 
\beta_{T}(\rho):=1-\Tr \rho T~~(\rho\in\mathcal{S}_1).
\end{equation}
Then, the function
\[
\begin{array}{cl}
\gamma_{T}(\rho):=1-\beta_{T}(\rho)=\mathrm{Tr}\rho T&~~~(\rho\in\mathcal{S}_1)
\end{array}
\]
is called the power of the test $T$.

A test with lower error probabilities is better, 
but type I and type II error probabilities can not be minimized simultaneously. 
Since there often exists a trade-off relation between type I error probability and type II error probability, accordingly, 
we take a permissible error constant $\alpha\in (0,1)$ for the first error probability, 
which is called the {\it (significance) level}.
Hence, we treat tests $T$ with level $\alpha$, 
i.e., $\Tr T \rho \le \alpha$ for all states $\rho \in \mathcal{S}_0$. 
Then, we denote the set of tests $T$ with level $\alpha$ by ${\cal T}_{\alpha}$,
i.e.,  
\begin{align}
{\cal T}_{\alpha}
:=\{T| 0\le T \le 1, ~ \Tr T \rho \le \alpha, \forall \rho \in \mathcal{S}_0 \} .
\end{align}
A test $T$ with level $\alpha$ is called a UMP test (Uniformly Most Powerful test)
when its type II error probability is the minimum among tests with level $\alpha$,
i.e., 
$ \beta_{T}(\rho) \le  \beta_{T'}(\rho) $ 
for all states $\rho \in \mathcal{S}_1$ and  
for all tests $T'\in {\cal T}_\alpha$.
A problem to derive a UMP test is often treated in quantum hypothesis testing,
however a UMP test may not exist when the null hypothesis $H_1$ is composite. 
Then, we need to modify this formulation.

A family of quantum states on a hypothesis testing problem often has parameters that are unrelated to the hypotheses. 
We call the unrelated parameters the disturbance parameters. 
For example, let consider the following hypothesis testing problem of the number parameter $N$ 
for a family of quantum Gaussian states $\{\rho_{\theta,N}\}_{\theta\in \C,N\in\R_{> 0}}$:
\[
H_0:N\le N_0~vs.~H_0:N> N_0,
\]
where $N_0$ is a positive constant. 
In this case, the disturbance parameter is the mean parameter $\theta \in \C$, 
which is unrelated to the number parameter $N$. 

This situation can be formulated in the following way.
It is assumed that our parameterized family is given as $\{\rho_{\theta,\xi}\}_{\theta\in \Theta,\xi \in \Xi}$,
in which, the parameter $\xi\in \Xi$ is the disturbance parameter 
and the parameter $\theta\in \Theta$ is related to our hypotheses.
In order to formulate our problem, we assume that
the parameter space $\Theta$ is given as the union of two disjoint subsets $\Theta_0$ and $\Theta_1$.
When $\mathcal{S}_i=\{\rho_{\theta,\xi} \}_{\theta\in \Theta_i,\xi\in \Xi}$ for $i=0,1$,
our problem and our two hypotheses $H_0$ and $H_1$ are described in the following way:
\[
H_0:\theta \in \Theta_0 ~vs.~ H_1:\theta \in\Theta_1
\hbox{ with }
\{\rho_{\theta,\xi}\}_{\theta\in \Theta,\xi \in \Xi}.
\]

The min-max criterion for the disturbance parameter $\xi\in \Xi$ is 
provided by the notion that 
it is better for a test to have smaller maximum value of 
the type II error probability among all disturbance parameters $\xi\in \Xi$.
Then, the optimal test with level $\alpha$
is given as the test $T_0$ with level $\alpha$ satisfying the following equation:
\[
\sup_{\xi \in \Xi}\beta_{T_0}(\rho_{\theta,\xi})
=\inf_{T\in {\cal T}_{\alpha} }\sup_{\xi \in \Xi}\beta_{T}(\rho_{\theta,\xi})~~~(\forall \theta \in \Theta_1).
\]
We call the above optimal test a \textbf{UMP min-max test} with level $\alpha$. 
Our main task is to derive a UMP test and a UMP min-max test 
for various hypothesis testing problems of the quantum Gaussian states in the following sections.

%%%The UMP min-max test coincides the UMP test when every disturbance space consist of one element, namely, $h$ is injective. Therefore the UMP min-max test is generalization of the UMP test. We will mainly derive the UMP test and the UMP min-max test in this paper.

However, it is often difficult to find a UMP min-max test or a UMP test.
In this case, we impose the additional condition on tests. 
A test $T$ with level $\alpha$ is called an \textbf{unbiased test} 
if the test satisfies $\beta_{T}(\rho_{\theta,\xi})\le 1-\alpha$ for 
$\theta \in \Theta_1$ and $\xi \in \Xi$.
Then, we sometimes seek the optimal test under the unbiasedness for tests. 

\begin{rem}
In the above formulation, the parametrization map $(\theta,\xi) \mapsto \rho_{\theta,\xi}$
is not necessarily injective.
There may exist a point $\theta_0\in \Theta$ such that
different $\xi\in\Xi$ provides the same state $\rho_{\theta,\xi}$.
Such a point $\theta_0\in \Theta$ is called a singular point,
Even if there exists a singular point, the above formulation works properly. 
\end{rem}

%%%%%%%%%%%%%%%%%%%%%%%%%%%%%%%%%%%%%%%%%%%%%%%%%%%%%%%%%%%%%%%%%%%%%
%%%%%%%%%%%%%%%%%%%%%%%%%%%%%%%%%%%%%%%%%%%%%%%%%%%%%%%%%%%%%%%%%%%%%
\subsection{Invariance of UMP min-max test}

~

For a simple derivation of the optimal test, we sometimes focus on
a unitary (projective) representation $V$ of a group $G$ on the Hilbert space $\mathcal{H}$. The unitary (projective) representation $V$ is called {\it covariant} concerning the disturbance parameter space $\Xi$,
when there is an action of group $G$ to the disturbance parameter space $\Xi$ such that
\begin{equation}\Label{cov}
V_g\rho_{\theta,\xi}V_g^*
=
\rho_{\theta,g\cdot \xi}, \quad \forall \theta \in \Theta, ~\forall \xi \in\Xi,~\forall g\in G. 
\end{equation}

Now, we impose the invariance for tests under the above covariance.
A test $T$ is called $G$ \textit{invariant test} concerning a (projective) representation $V$
if $V_g T V_g^*=T$ holds for any $g\in G$. 
A \textit{UMP invariant test} is defined by 
an invariant test with the minimum type II error in the class of tests with level $\alpha$. 
That is, an invariant test $T$ with level $\alpha$ is called 
a UMP invariant test with level $\alpha$
when $ \beta_T(\rho) \le  \beta_{T'}( \rho) $ holds
for all states $\rho \in \mathcal{S}_1$ and  
for all invariant tests $T'\in {\cal T}_{\alpha}$.
Then, it is nothing but the optimal test in invariant tests. 

It is often easy to optimize the invariant test in virtue of the invariance and accordingly derive 
the UMP invariant test. 
Quantum Hunt-Stein theorem guarantees that
a UMP min-max test of level $\alpha$ is given as a UMP invariant test.
Quantum Hunt-Stein theorem for a compact group was given by Holevo\cite{Hol,Hol1} and quantum Hunt-Stein theorem for a non-compact case was shown by Bogomolov\cite{Bog} and Ozawa\cite{Oza} 
but not stated in the context of the quantum hypothesis testing. In the following, we restate quantum Hunt-Stein theorem as 
a theorem concerning the following hypothesis testing with the min-max criterion:
\begin{eqnarray}\label{C}
H_0:\theta \in\Theta_0~vs.~H_1:\theta \in\Theta_1
\hbox{ with }
\{\rho_{\theta,\xi}\}_{\theta\in \Theta,\xi \in \Xi}
\end{eqnarray}
for a family of quantum states $\{\rho_{\theta,\xi}\}_{\theta\in\C,\xi \in \Xi}$ on $\mathcal{H}$ with the disturbance parameter $\xi$.

At first, 
we consider the case that a compact group acts on a family of quantum states 
and derive the relation between a UMP min-max test and a UMP invariant test.

\begin{thm}$(\textbf{Quantum Hunt-Stein theorem for a compact group})\Label{cpt.H-S}$~
When a (projective) representation of a compact group $G$ 
satisfies 
the covariant condition $(\ref{cov})$ concerning the disturbance parameter space $\Xi$,
the following equations hold.
\begin{align*}
\inf_{\tilde{T} \in {\cal T}_{\alpha,V}}\sup_{\xi \in \Xi} \beta_{\tilde{T}}(\rho_{\theta,\xi})
&=\inf_{T\in {\cal T}_{\alpha}} \sup_{\xi \in \Xi}\beta_{T}(\rho_{\theta,\xi}) 
\quad(\forall \theta \in\Theta_1), \\
\inf_{\tilde{T} \in {\cal T}_{\alpha,u,V}}\sup_{\xi \in \Xi} \beta_{\tilde{T}}(\rho_{\theta,\xi})
&=\inf_{T\in {\cal T}_{\alpha,u}} \sup_{\xi \in \Xi}\beta_{T}(\rho_{\theta,\xi}) 
\quad(\forall \theta \in\Theta_1),
\end{align*}
where 
${\cal T}_{\alpha,V}$, 
${\cal T}_{\alpha,u}$, and
${\cal T}_{\alpha,u,V}$ 
are
the set of tests of level $\alpha$
that are invariant concerning the (projective) representation $V$,
the set of unbiased tests of level $\alpha$,
and
the set of unbiased tests of level $\alpha$
that are invariant concerning the (projective) representation $V$.
\end{thm}

Theorem \ref{cpt.H-S} yields the following proposition which insists that an invariant UMP test is also a UMP min-max test.

\begin{prp}\Label{cpt.inv.min-max}
Assume that a unitary (projective) representation of a compact group $G$ satisfies 
the covariant condition $(\ref{cov})$ concerning the disturbance parameter space $\Xi$.
If there exists a UMP (unbiased) invariant test $T_0$ with level $\alpha$ for hypothesis testing $(\ref{C})$,
$T_0$ is also a UMP (unbiased) min-max test with level $\alpha$ for 
the disturbance parameter $\xi \in \Xi$.
\end{prp}

Next, we state quantum Hunt-Stein theorem for a non-compact case,
which requires the amenability for a group $G$. The definition and several properties of an amenable group are written in \cite{G}. The following proposition is known.

\begin{prp}\Label{AIPM}\cite{B-P}~
For a locally compact Hausdorff group $G$, 
there exists an asymptotic probability measure on $G$ if and only if $G$ is amenable.
\end{prp}

Here, an asymptotic probability measure $\{\nu_n\}$ on $G$ is defined by a sequence of probability measures on $G$ satisfying
\[
\displaystyle\lim_{n\to\infty}|\nu_n(g\cdot B)-\nu_n(B)|=0
\]
for any $g\in G$ and any Borel set $B\subset G$. For example, all of compact Lie groups like $(S^1)^n$ are amenable groups. A finite-dimensional Euclid space $\R^n$ and the whole of integers are amenable groups too. 

In addition, we need the completeness of a family of quantum states. 
A family of quantum states is called complete if for any bounded linear operator $X$, the following holds:
\[
\mathrm{Tr}(\rho X)=0~(\forall\rho\in\mathcal{S})~~~is~equivalent~with ~~~X=0.
\]
The whole of pure states on a finite-dimensional quantum system and 
the quantum Gaussian states $\{\rho_{\theta,N}\}_{\theta\in\C}$ 
have the completeness for any $N>0$. 

\begin{thm}$($\textbf{Quantum Hunt-Stein theorem} \cite{Bog},\cite{Oza}$)$\Label{H-S}~
Let $\mathcal{S}$ be a complete family of quantum states on a Hilbert space $\mathcal{H}$ of at most countable dimension. 

When a unitary (projective) representation of an amenable group $G$ 
satisfies 
the covariant condition $(\ref{cov})$ concerning the disturbance parameter space $\Xi$,
the following equation holds.
\begin{align*}
\inf_{\tilde{T} \in {\cal T}_{\alpha,V}}\sup_{\xi \in \Xi} \beta_{\tilde{T}}(\rho_{\theta,\xi})
&=\inf_{T\in {\cal T}_{\alpha}} \sup_{\xi \in \Xi}\beta_{T}(\rho_{\theta,\xi}) 
\quad(\forall \theta \in\Theta_1), \\
\inf_{\tilde{T} \in {\cal T}_{\alpha,u,V}}\sup_{\xi \in \Xi} \beta_{\tilde{T}}(\rho_{\theta,\xi})
&=\inf_{T\in {\cal T}_{\alpha,u}} \sup_{\xi \in \Xi}\beta_{T}(\rho_{\theta,\xi}) 
\quad(\forall \theta \in\Theta_1).
\end{align*}
\end{thm}

From Theorem \ref{H-S}, the following proposition which insists that a UMP test is also a UMP min-max test is obtained. 

\begin{prp}\Label{inv.min-max}

Let $\mathcal{S}$ be a complete family of quantum states on a Hilbert space $\mathcal{H}$ of at most countable dimension. 
Assume that a unitary (projective) representation of an amenable group $G$ satisfies 
the covariant condition $(\ref{cov})$ concerning the disturbance parameter space $\Xi$.
If there exists a UMP (unbiased) invariant test $T_0$ with level $\alpha$ for hypothesis testing $(\ref{C})$,
$T_0$ is also a UMP (unbiased) min-max test with level $\alpha$ for 
the disturbance parameter $\xi \in \Xi$.

\end{prp}

We give a proof of quantum Hunt-Stein theorem on the hypothesis testing in appendix. 
Since our proof is restricted to the situation of the hypothesis testing, 
it is plainer than Bogomolov's proof. 
When we additionally impose the unbiasedness to our tests,
the arguments similar to Theorems \ref{cpt.H-S} and \ref{H-S} hold.
Therefore, Propositions \ref{cpt.inv.min-max} and \ref{inv.min-max} insist that 
an UMP unbiased invariant test is also an UMP unbiased min-max test.
%In the later sections, using Theorems \ref{cpt.H-S} and \ref{H-S} or their unbiased versions, we will treat the hypothesis testing problems about the quantum Gaussian states.

\begin{rem}
The above theorems can be restated with a general action of a group $G$
on the space $\mathcal{B}(\mathcal{H})$ of bounded operators on $\mathcal{H}$.
Indeed, when the action comes from a unitary (projective) representation $V$ as $g\cdot T:=V_g T V_g^*$,
the compatibility condition 
\begin{equation}\Label{compatibility}
\mathrm{Tr}(g\cdot\rho)X=\mathrm{Tr}\rho(g^{-1}\cdot X)
\end{equation}
holds.
When an action of an amenable group $G$
on the space $\mathcal{B}(\mathcal{H})$
satisfies the condition (\ref{compatibility}), it is called compatible.
Even if the (projective) representation $V$ is replaced by a compatible group action on the space $\mathcal{B}(\mathcal{H})$,
we can show Theorems \ref{cpt.H-S} and \ref{H-S} and Propositions \ref{cpt.inv.min-max} and \ref{inv.min-max}
in the same way.
\end{rem}

%There are few studies for the composite hypothesis testing problem (\cite{Hay1}), and especially, the hypothesis testing problem for quantum Gaussian 
%states is treated only in the simple hypotheses case (\cite{Mos}). The quantum Hunt-Stein theorem is effective for the composite hypothesis testing 
%problem with symmetry, and applied for the composite hypothesis testing problem of quantum Gaussian states in the later sections. 

\subsection{Reduction methods on quantum hypothesis testing problems}~

Usually, a meaningful quantum hypothesis testing problem has a complicated structure.
Hence, it is needed to simplify such a given quantum hypothesis testing problem.
For this kind of simplification, we prepare some lemmas reducing a quantum hypothesis testing problem to a fundamental one. 
Thereafter, we summed up those lemmas as Theorem $\ref{thm3}$. Those are very useful and frequently used in later sections. 
Let $\mathcal{H}_1$ be a Hilbert space with at most countable dimension, and $\{\rho_{\theta,\xi_1}\}_{\theta\in\Theta, \xi_1\in\Xi_1}$ be a family of quantum states on $\mathcal{H}_1$. 

The following lemma shows that we can transform a quantum system in a testing problem by a unitary transformation.  

\begin{lem}\Label{lem1}
The following equivalent relation holds for any unitary operator $U$ on $\mathcal{H}_1$.
A test $T$ is a UMP (unbiased, min-max) test with level $\alpha$ for
\begin{align}\Label{lem1-1}
H_0:\theta \in \Theta_0 ~vs.~ H_1:\theta \in\Theta_1
\hbox{ with }
\mathcal{S}=\{\rho_{\theta,\xi_1} \}_{\theta\in \Theta,\xi_1\in \Xi_1}
\end{align}
if and only if
a test $U^* T  U$ is a UMP (unbiased, min-max) test with level $\alpha$ for
\begin{align}\Label{lem1-2}
H_0:\theta \in \Theta_0 ~vs.~ H_1:\theta \in\Theta_1
\hbox{ with }
\{U^* \rho_{\theta,\xi_1} U \}_{\theta\in \Theta,\xi_1\in \Xi_1}.
\end{align}
\end{lem}

The above lemma is obvious.  

The following two lemma shows that we can erase a partial system from which we can not obtain the information about the testing parameter.

\begin{lem}\Label{lem2}
Let $\mathcal{H}_2$ be a Hilbert space with at most countable dimension, and $\{\rho_{\xi_2}\}_{\xi_2\in\Xi_2}$ be a family of quantum states on $\mathcal{H}_2$. Then, the following equivalent relation holds.

A test $T$ is a UMP (unbiased, min-max) test with level $\alpha$ for
\begin{align}\Label{lem2-1}
H_0:\theta \in \Theta_0 ~vs.~ H_1:\theta \in\Theta_1
\hbox{ with }
\{\rho_{\theta,\xi_1} \}_{\theta\in \Theta,\xi_1\in \Xi_1}
\end{align}
if and only if
a test $T \otimes I$ is a UMP (unbiased, min-max) test with level $\alpha$ for
\begin{align}\Label{lem2-2}
H_0:\theta \in \Theta_0 ~vs.~ H_1:\theta \in\Theta_1
\hbox{ with }
\{\rho_{\theta,\xi_1}\otimes \rho_{\xi_2}\}_{\theta\in \Theta,\xi_1\in \Xi_1,\xi_2\in \Xi_2}.
\end{align}
\end{lem}

\begin{proof}

We assume that $T$ is a UMP (unbiased) test with level $\alpha$ for $(\ref{lem2-1})$. We define a completely positive map $\Lambda_{\rho_{\xi_2}}$ for each $\xi_2\in\Xi_2$ from the space of bounded self-adjoint operators on $\mathcal{H}\otimes\mathcal{H}'$ to that on $\mathcal{H}$ as 
\[
\mathrm{Tr}(\Lambda_{\rho_{\xi_2}}(X)\rho) = \mathrm{Tr}(X(\rho\otimes\rho_{\xi_2}))
\]
holds for any trace class operator $\rho$ on $\mathcal{H}$. Let $T'$ be an arbitrary test with level $\alpha$ for $(\ref{lem2-2})$. Then, $0\le\Lambda_{\rho_{\xi_2}}(T')\le I$ holds by the completely positivity of $\Lambda_{\rho_{\xi_2}}$. Moreover, since $T'$ is a test with level $\alpha$, 
\[
\mathrm{Tr}(\Lambda_{\rho_{\xi_2}}(T')\rho_{\theta,\xi_1})=\mathrm{Tr}(T'(\rho_{\theta,\xi_1}\otimes\rho_{\xi_2}))\le \alpha
\]
holds for any $\theta\in\Theta_0$, that is, $\Lambda_{\rho_{\xi_2}}(T')$ is a test with level $\alpha$ for $(\ref{lem2-1})$. Similarly, if $T'$ satisfies unbiasedness, $\Lambda_{\rho_{\xi_2}}(T')$ satisfies unbiasedness too. Therefore, we get
\[
\beta_{T\otimes I}(\rho_{\theta,\xi_1}\otimes\rho_{\xi_2})
=\beta_{T}(\rho_{\theta,\xi_1})
\le \beta_{\Lambda_{\rho_{\xi_2}}(T')}(\rho_{\theta,\xi_1})
=\beta_{T'}(\rho_{\theta,\xi_1}\otimes\rho_{\xi_2})
\]
for any $\theta\in\Theta_1,\xi_1\in\Xi_1,\xi_2\in\Xi_2$. Thus, $T\otimes I$ is a UMP (unbiased) test with level $\alpha$ for $(\ref{lem2-2})$.

Similarly, when $T$ is a UMP (unbiased) min-max test with level $\alpha$ for $(\ref{lem2-1})$, $T \otimes I$ is a UMP (unbiased) min-max test with level $\alpha$ for $(\ref{lem2-2})$. 

The converse is obvious.

\end{proof}

\begin{lem}\Label{lem3}

Let $\mathcal{H}_3$ be a Hilbert space with at most countable dimension, and $\{\rho_{\theta,\xi_1,\xi_3}\}_{\theta\in\Theta, \xi_1\in\Xi_1, \xi_3\in\Xi_3}$ be a family of quantum states on $\mathcal{H}_3$. We assume that $V$ is an irreducible unitary (projective) representation of an amenable group $G$ on $\mathcal{H}_3$ and covariant concerning the disturbance parameter space $\Xi_1\times\Xi_3$. In addition, if $G$ is non-compact, we assume that the family of quantum states $\{\rho_{\theta,\xi_1}\otimes \rho_{\theta,\xi_1,\xi_3} \}_{\theta\in \Theta,\xi_1\in \Xi_1,\xi_3\in \Xi_3}$ is complete. Then, the following equivalent relation holds. 

A test $T$ is a UMP (unbiased) min-max test with level $\alpha$ for
\begin{align}\Label{lem3-1}
H_0:\theta \in \Theta_0 ~vs.~ H_1:\theta \in\Theta_1
\hbox{ with }
\{\rho_{\theta,\xi_1} \}_{\theta\in \Theta,\xi_1\in \Xi_1}
\end{align}
if and only if
a test $T \otimes I$ is a UMP (unbiased) min-max test with level $\alpha$ for
\begin{align}\Label{lem3-2}
H_0:\theta \in \Theta_0 ~vs.~ H_1:\theta \in\Theta_1
\hbox{ with }
\{\rho_{\theta,\xi_1}\otimes \rho_{\theta,\xi_1,\xi_3} \}_{\theta\in \Theta,\xi_1\in \Xi_1,\xi_3\in \Xi_3}.
\end{align}
\end{lem}

\begin{proof}
When $T$ is a UMP (unbiased) min-max test with level $\alpha$ for $(\ref{lem3-1})$, we show that $T\otimes I$ is a UMP (unbiased) min-max test with level $\alpha$ for $(\ref{lem3-2})$. Let $T'$ be a test with level $\alpha$ for $(\ref{lem3-2})$. We can assume that $T'$ is invariant concerning to the irreducible (projective) representation $V$ due to quantum Hunt-Stein theorem (Theorem \ref{cpt.H-S} or Theorem \ref{H-S}). Thus, $T'$ is represented as $T''\otimes I$ by using the test $T''$ on $\mathcal{H}$ due to Schur's lemma. Then, since $T'$ is a test with level $\alpha$ for $(\ref{lem3-2})$,  the test $T''$ is a test with level $\alpha$ for $(\ref{lem3-1})$. Therefore, the type II error probability of $T\otimes I$ is less than or equal to that of $T'=T''\otimes I$.

The converse is obvious.

\end{proof}

We get the following theorem by summing up the above three lemmas.

\begin{thm}\Label{thm3}
Let $\mathcal{H}_1,\mathcal{H}_2,\mathcal{H}_3$ be Hilbert spaces with at most countable dimension, and $\mathcal{S}_1:=\{\rho_{\theta,\xi_1}\}_{\theta\in\Theta, \xi_1\in\Xi_1},\mathcal{S}_2:=\{\rho_{\xi_2}\}_{\xi_2\in\Xi_2},\mathcal{S}_3:=\{\rho_{\theta,\xi_1,\xi_3}\}_{\theta\in\Theta, \xi_1\in\Xi_1, \xi_3\in\Xi_3}$ be families of quantum states on $\mathcal{H}_1,\mathcal{H}_2,\mathcal{H}_3$, respectively. We assume that $V$ is an irreducible (projective) representation of an amenable group $G$ on $\mathcal{H}_3$ and covariant concerning the disturbance parameter space $\Xi_1\times\Xi_3$ of the family of quantum states $\{\rho_{\theta,\xi_1,\xi_3}\}_{\theta\in\Theta, \xi_1\in\Xi_1,  \xi_3\in\Xi_3}$. In addition, if $G$ is non-compact, we assume that the family of quantum states $\{\rho_{\theta,\xi_1}\otimes \rho_{\theta,\xi_1,\xi_3} \}_{\theta\in \Theta,\xi_1\in \Xi_1,\xi_3\in \Xi_3}$ is complete. Then, the following equivalent relation holds. 

A test $T$ is a UMP (unbiased) min-max test with level $\alpha$ for
\begin{align}\Label{thm3-1}
H_0:\theta \in \Theta_0 ~vs.~ H_1:\theta \in\Theta_1
\hbox{ with }
\{\rho_{\theta,\xi_1} \}_{\theta\in \Theta,\xi_1\in \Xi_1}
\end{align}
if and only if
a test $U^*(T \otimes I_{\mathcal{H}_2}\otimes I_{\mathcal{H}_3})U$ is a UMP (unbiased) min-max test with level $\alpha$ for
\begin{eqnarray}\Label{thm3-2}
&H_0:\theta \in \Theta_0 ~vs.~ H_1:\theta \in\Theta_1&\nonumber\\
&\hbox{ with }
\{U^*(\rho_{\theta,\xi_1}\otimes \rho_{\xi_2}\otimes\rho_{\theta,\xi_1, \xi_3})U \}_{\theta\in \Theta,\xi_1\in \Xi_1,\xi_2\in \Xi_2,\xi_3\in \Xi_3},&
\end{eqnarray}
where $U$ is a unitary operator on $\mathcal{H}_1\otimes\mathcal{H}_2\otimes\mathcal{H}_3$.

In addition, when $\mathcal{H}_3$ is the empty set $\phi$, a test $T$ is a UMP (unbiased, min-max) test with level $\alpha$ for $(\ref{thm3-1})$
if and only if
a test $U'^*(T \otimes I_{\mathcal{H}_2})U'$ is a UMP (unbiased, min-max) test with level $\alpha$ for
\begin{align}\Label{thm3-3}
H_0:\theta \in \Theta_0 ~vs.~ H_1:\theta \in\Theta_1
\hbox{ with }
\{U'^*(\rho_{\theta,\xi_1}\otimes \rho_{\xi_2})U' \}_{\theta\in \Theta,\xi_1\in \Xi_1,\xi_2\in \Xi_2},
\end{align}
where $U'$ is a unitary operator on $\mathcal{H}_1\otimes\mathcal{H}_2$.
\end{thm} 

We considered some reduction methods on quantum hypothesis testing problems. 
Those results are effectively used in testing problems for quantum Gaussian states. Note that those reduction methods are not peculiar to quantum hypothesis testing. Indeed, when we construct optimal tests in the testing problems for Gaussian distributions, we perform the suitable orthogonal transformation (called the Helmert transformation) to the random variables distributed according to the Gaussian distributions, and erase useless parts in the transformed random variables. By those reduction, the testing problems for Gaussian distributions are transformed to fundamental problems, and optimal tests for the fundamental problems are constructed by using $\chi^2$, $t$ and $F$ distribution. In later sections, we treat testing problems for quantum Gaussian states. Then, we transform the testing problems to some fundamental problems by using the above reduction methods, and derive optimal tests for all testing problems by constructing optimal tests for fundamental problems.

%In the construction of the optimal tests for those testing problems, essential point is the reduction to established form is essential to construct  

%%%%%%%%%%%%%%%%%%%%%%%%%%%%%%%%%%%%%%%%%%%%%%%%%%%%%%%%%%%%%%%%
%%%%%%%%%%%%%%%%%%%%%%%%%%%%%%%%%%%%%%%%%%%%%%%%%%%%%%%%%%%%%%%%
\section{Relation between classical and quantum hypothesis testing}\Label{classical result}

~

In this section, we describe the relation between classical hypothesis testing and quantum hypothesis testing of the commutative case, 
and summed up some fundamental facts in classical hypothesis testing. 
Thereafter, we treat a quantum counterpart of $\chi^2$ and $F$ tests.

%%%%%%%%%%%%%%%%%%%%%%%%%%%%%%%%%%%%%%%%%%%%%%%%%%%%%%%%%%%%%%%%%%%%%
%%%%%%%%%%%%%%%%%%%%%%%%%%%%%%%%%%%%%%%%%%%%%%%%%%%%%%%%%%%%%%%%%%%%%
\subsection{Hypothesis testing for a commutative quantum states family}~

We consider a hypothesis testing for a family of commutative quantum states. Let $\mathcal{H}$ be a Hilbert space with at most countable dimension, and $\mathcal{S}=\{\rho_{\theta,\xi}\}_{\theta\in\Theta, \xi\in\Xi}$ be a family of commutative quantum states on  $\mathcal{H}$. By using a suitable orthonormal basis $\{|x\rangle\}_{x\in\mathcal{X}}$ of $\mathcal{H}$, each state of $\{\rho_{\theta,\xi}\}_{\theta\in\Theta, \xi\in\Xi}$ is represented as
\[
\rho_{\theta,\xi}=\displaystyle\sum_{x\in\mathcal{X}}p_{\theta,\xi}(x)|x\rangle\langle x|,
\]
where $\{p_{\theta,\xi}\}_{\theta\in\Theta, \xi\in\Xi}$ is a family of probability distribution on $\mathcal{X}$. Therefore,  a family of commutative quantum states $\mathcal{S}$ is identified with a family of probability distribution whose each probability is composed by the eigenvalues of the quantum states in $\mathcal{S}$.

Then, for an arbitrary test operator $T$, when we define the diagonal part $T_{diag}$ of the test by a test operator $\sum_{x\in\mathcal{X}} \langle x|T|x\rangle |x\rangle\langle x|$, the diagonal part satisfies $\mathrm{Tr}\rho_{\theta,\xi}T=\mathrm{Tr}\rho_{\theta,\xi}T_{diag}$. In particular, the type I and type II error probabilities satisfy $\alpha_{T}(\theta)=\alpha_{T_{diag}}(\theta)$ and $\beta_{T}(\theta)=\beta_{T_{diag}}(\theta)$. 
Since the efficiency of an arbitrary test can be achieved by the diagonal part of the test, 
our tests can be restricted to tests which have the form $\sum_{x\in\mathcal{X}} t_x |x\rangle\langle x|~(t_x\in[0,1])$. Therefore, a test operator $T=\sum_{x\in\mathcal{X}} t_x |x\rangle\langle x|$ for a family of commutative quantum states is identified with a test function $\ph_T(x):=t_x\in[0,1]$ on $\mathcal{X}$ which is composed by the eigenvalues of the test $T$. 

As stated above , the formulation of quantum hypothesis testing for a family of commutative quantum states is identified with the formulation of classical hypothesis testing for a family of probability distributions. Therefore, classical hypothesis testing theory is often effective to consider some kind of problems in quantum hypothesis testing. 

In classical hypothesis testing, problems for the exponential family is intensively studied, and results for those problems is necessary in hypothesis testing for quantum Gaussian states. Here, an $l$-parameter exponential family is defined as a parametric family of probability distributions $\mathcal{P}=\{p_{\theta}\}_{\theta\in\Theta\subset\R^l}$ 
($\Theta\subset\R^l$ is a convex set) 
on a measurable set $\Omega$ such that all probability distributions are represented as 
\[
p_{\theta}(\omega)=\mathrm{exp}\left(F(\omega)+\displaystyle\sum_{k=1}^l X_l(\omega)\theta_l-\psi(\theta)\right)
\]
where $F, X_1,\cdot\cdot\cdot, X_l$ are $\R$-valued random variables not depending on parameters $\theta=(\theta_1,\cdot\cdot\cdot,\theta_l)$, and $\psi:\Theta\to\R$ is a function \cite{A-N}. In the above exponential family, the parameters $\theta=(\theta_1,\cdot\cdot\cdot,\theta_l)$ are called natural parameters. 

In the following subsections, we review fundamental results in classical hypothesis testing, and derive optimal tests for two problems on quantum Gaussian states by using the results.

\subsection{Hypothesis testing with one-parameter exponential family}~

In order to treat a typical commutative case,
we employ the following known theorem for a one-parameter exponential family. 

\begin{thm}\cite{Leh2}\Label{cla1}
Let $\mathcal{P}=\{p_{\theta}\}_{\theta\in\Theta}$ be a one-parameter exponential family with a natural parameter $\theta$. 
Then, for any real number $C\in\Theta$, there exist 
constants $\gamma\in (0,1]$, $K\in\R$ and a test function
\[
\ph(\omega)=\left\{
\begin{array}{ll}
 0&\left( X_1(\omega)<K\right)\\
 \gamma&\left( X_1(\omega)=K\right)\\
 1&\left( X_1(\omega)>K\right)
\end{array}\right. 
\]
satisfying 
\begin{equation}\Label{alpha.const.0}
\mathrm{E}_{p_{C}}[\ph(X_1)]=\alpha,
\end{equation}
where $\mathrm{E}_{p_{C}}[~]$ means an expectation under the distribution $p_{C}$.
The test function $\ph$ is a UMP test function for the hypothesis testing problem
\[
H_0:\theta\le C~vs.~H_1:\theta>C,
\hbox{ with }
\{p_{\theta}\}_{\theta \in \Theta}.
\]
\end{thm}

Next, we apply the above theorem to a quantum hypothesis testing problem. 
When we define a one-parameter exponential family $\{p_{n,w}\}_{w\in\R_{<0}}$ on $\Z_{\ge 0}^n$ by
\[
p_{n,w}(k)
 =\mathrm{exp}\left\{X_1(k)w-\psi(w)\right\}
\]
where $k=(k_1,\cdot\cdot\cdot,k_n)\in\Z_{\ge 0}^n, X_1(k):=\sum k_j,~ \psi(w):=-n\mathrm{log}(1-\mathrm{e}^{w})$, the family of quantum Gaussian states $\{\rho_{0,N}^{\otimes n}\}_{N\in\R_{>0}}$ is identified with the exponential family $\{p_{n,w}\}_{w\in\R_{>0}}$ since
\[
\rho_{0,N}^{\otimes n}
=\displaystyle\sum_{k\in\Z_{\ge 0}^n}p_{n,\mathrm{log}\left(\frac{N}{N+1}\right)}(k)|k \rangle\langle k|
\]
holds. Thus, a quantum hypothesis testing problem
\[
H_0:N\in (0,N_0]~vs.~H_1:N\in(N_0,\infty)
\hbox{ with }
\{\rho_{0,N}^{\otimes n}\}_{N\in\R_{> 0}} \eqno{(\mathrm{H}\textrm{-}\mathrm{\chi^2})}
\]
is translated to a classical hypothesis testing problem
\[
H_0:w\in \left(-\infty,\mathrm{log}\frac{N_0}{1+N_0}\right]~vs.~H_1:w\in \left(\mathrm{log}\frac{N_0}{1+N_0},0\right)
\hbox{ with }
\{p_{n,w}\}_{w\in\R_{< 0}}.
\]

Since $\{p_{n,w}\}_{w\in\R_{<0}}$ is a one-parameter exponential family with the natural parameter $w$, 
the following test function $\ph$ is a UMP test for the above hypothesis testing problem due to Theorem \ref{cla1}. The test function $\ph$ is defined by
\begin{equation}\Label{testfunction1}
\ph(k):=\left\{
\begin{array}{ll}
 0&( X_1(k)<K_0)\\
 \gamma&( X_1(k)=K_0)\\
 1&( X_1(k)>K_0),
\end{array}\right. 
\end{equation}
where $X_1(k):=\sum_{j=1}^n k_j$ for $k=(k_1,\cdot\cdot\cdot,k_n)\in\Z_{\ge0}^n$ and,
the constants $K_0\in\Z_{\ge 0}$ and $\gamma\in(0,1]$ are uniquely determined by 
\begin{eqnarray}\Label{alpha.const.3"}
&1-\displaystyle\sum_{K=0}^{K_0}\binom{K+n-1}{n-1}\left(\frac{1}{N_0+1}\right)^n \left(\frac{N_0}{N_0+1}\right)^K
< \alpha& \nonumber\\
&\le 1-\displaystyle\sum_{K=0}^{K_0-1}\binom{K+n-1}{n-1}\left(\frac{1}{N_0+1}\right)^n \left(\frac{N_0}{N_0+1}\right)^K,&\\
&\gamma:=\frac{\alpha-\left(1-\displaystyle\sum_{K=0}^{K_0}\binom{K+n-1}{n-1}\left(\frac{1}{N_0+1}\right)^n \left(\frac{N_0}{N_0+1}\right)^K\right)}
{\binom{K_0+n-1}{n-1}\left(\frac{1}{N_0+1}\right)^n \left(\frac{N_0}{N_0+1}\right)^{K_0}}&
\end{eqnarray} 
which corresponds to the condition $(\ref{alpha.const.0})$ in Theorem \ref{cla1}. Therefore,  we get the following theorem.

\begin{thm}\Label{chi2 test}
For the hypothesis testing problem $(\mathrm{H}\textrm{-}\mathrm{\chi^2})$, the test 
\[
T_{\alpha,N_0}^{[\mathrm{\chi^2}],n}:=\displaystyle\sum_{k\in\Z_{\ge 0}^n}\ph(k)|k\rangle\langle k|.
\]
 is a UMP test with level $\alpha$.
\end{thm}

We constructed a test $T_{\alpha,N_0}^{[\mathrm{\chi^2}],n}$ for the hypothesis testing problem $(\mathrm{H}\textrm{-}\mathrm{\chi^2})$. In the classical system, the hypothesis testing problem $(\mathrm{H}\textrm{-}\mathrm{\chi^2})$ corresponds to
\[
H_0:N\in(0,N_0]~vs.~H_1:N\in (N_0,\infty) ~\hbox{with}~\{\mathrm{G}_{0,N}^{(n)}\}_{N\in\R_{>0}}
\]
where $\mathrm{G}_{0,N}$ is a Gaussian distribution with the mean $0$ and the variance $N$. A $\chi^2$ test is a UMP test for the above classical hypothesis testing problem. Therefore, the test $T_{\alpha,N_0}^{[\mathrm{\chi^2}],n}$ can be regarded as a quantum counterpart of the $\chi^2$ test, and we call the test $T_{\alpha,N_0}^{[\mathrm{\chi^2}],n}$ a quantum $\chi^2$ test with $n$ degrees of freedom and level $\alpha$.

%%%%%%%%%%%%%%%%%%%%%%%%%%%%%%%%%%%%%%%%%%%%%%%%%%%%%%%%%%%%%%%%%%
\subsection{Hypothesis testing with a two-parameter exponential family}~

In order to treat another typical commutative case,
we employ the following known theorem for a two-parameter exponential family. 

\begin{thm}\cite{Leh2}\Label{cla2}~
Let $\{p_{\theta}\}_{\theta\in\Theta}$ be a two-parameter exponential family with a natural parameter $\theta=(\theta_1,\theta_2)$. 
Then, 
for a real number $C$ in the range of the first parameter,
there exist functions $\gamma:\R \to (0,1]$, $c_1,c_2:\R\to\R$ and a function
\[
\ph(\omega)=
\left\{
\begin{array}{ll}
0&(c_1(X_2(\omega))<X_1(\omega)<c_1(X_2(\omega)))\\
\gamma(X_2(\omega))&(X_1(k)=c_1(X_2(\omega))~or~c_2(X_2(\omega)))\\
1&(X_1(\omega)<c_1(X_2(\omega))~or~c_2(X_2(\omega))<X_1(\omega)),
\end{array}
\right.
\]
satisfying 
\begin{eqnarray}
\mathrm{E}_{p_{C}}[\ph(X_1,X_2)|X_2=x_2]&=&\alpha,\Label{cla.const}\\
\mathrm{E}_{p_{C}}[X_1\ph(X_1,X_2)|X_2=x_2]&=&\alpha\mathrm{E}_{p_{C}}[X_1|X_2=x_2]\Label{cla.const2}
\end{eqnarray}
for any $x_2\in\R$, where $\mathrm{E}_{p_{C}}[~|X_2=x_2]$ means a conditional expectation. 
Then, $\ph$ is a UMP unbiased test function for the hypothesis testing problem
\[
H_0:\theta_1=C~vs.~H_1:\theta_1\ne C,
\hbox{ with }
\{p_{\theta_1,\theta_2}\}_{(\theta_1,\theta_2)\in\Theta}.
\]

\end{thm}

Next, we apply the above theorem to a quantum hypothesis testing problem. When we define a two-parameter exponential family $\{p_{n,u,v}\}_{u\in\R, v\in\R_{<0}}$ on $\Z_{\ge 0}^{m+n}$ by 
\begin{equation}\Label{e-family}
p_{m,n,u,v}(k,l):=\mathrm{exp}\left\{X_1(k)u+X_2(k,l)v-\psi(u,v)\right\}
\end{equation}
where $X_1(k):=\sum_{i=1}^m k_i, X_2(k,l):=\sum_{i=1}^{m}k_i + \sum_{j=1}^{n}l_j,  \psi(u,v):=-m\mathrm{log}(1-\mathrm{e}^{u+v})-n\mathrm{log}(1-\mathrm{e}^v)$,
the family of quantum Gaussian states $\{\rho_{0,M}^{\otimes m}\otimes\rho_{0,N}^{\otimes n}\}_{M,N \in \R_{>0}}$ is identified with the family of  probability distributions $\{p_{n,u,v}\}_{u\in\R ,v\in\R_{<0}}$ since
\[
\begin{array}{lll}
\rho_{0,M}^{\otimes m}\otimes\rho_{0,N}^{\otimes n}
=&\displaystyle\sum_{k\in\Z_{\ge 0}^m, l\in\Z_{\ge 0}^n}& p_{m,n,\mathrm{log}\frac{M}{M+1}-\mathrm{log}\frac{N}{N+1}, \mathrm{log}\frac{N}{N+1}}(k,l)|k,l\rangle\langle k,l|
\end{array}
\]
holds. Thus, a quantum hypothesis testing problem
\[
H_0:M=N~vs.~H_1:M\ne N
\hbox{ with }
\{\rho_{0,M}^{\otimes m}\otimes\rho_{0,N}^{\otimes n}\}_{M,N\in\R_{> 0}} 
\eqno{(\mathrm{H}\textrm{-}\mathrm{F})}
\]
is translated to a classical hypothesis testing problem
\[
H_0:u=0~vs.~H_1:u\ne 0
\hbox{ with }
\{p_{n,u,v}\}_{u\in\R,v\in\R_{< 0}}.
\]

Since $\{p_{m,n,u,v}\}_{u\in\R,v\in\R_{\ge 0}}$ is a 
two-parameter exponential family with the natural parameters $u$ and $v$,
the following test function $\ph$ is a UMP unbiased test for the above hypothesis testing problem due to Theorem \ref{cla2}. The test function $\ph$ is defined by
\begin{equation}\Label{testfunction2}
\ph(k,l):=
\left\{
\begin{array}{ll}
0&(c_1(X_2(k,l))<X_1(k)<c_1(X_2(k,l)))\\
\gamma(X_2(k,l))&(X_1(k)=c_1(X_2(k,l))~or~c_2(X_2(k,l)))\\
1&(X_1(k)<c_1(X_2(k,l))~or~c_2(X_2(k,l))<X_1(k)),
\end{array}
\right.
\end{equation}
where
the functions $c_1,c_2: \Z_{\ge 0}\to\Z_{\ge 0} $, and $\gamma:\Z_{\ge 0}\to [0,1)$ 
are uniquely determined so that the relations
\begin{align}
&\gamma(s)\left(\binom{c_1(s)+m-1}{m-1}(1-\delta_{c_1(s),c_2(s)})+\binom{c_2(s)+m-1}{m-1}\right) \nonumber \\
&\quad +\displaystyle\sum_{j=0}^{c_1(s)-1}\binom{j+m-1}{m-1}j+\displaystyle\sum_{j=c_2(s)+1}^{\infty}\binom{j+m-1}{m-1}j\nonumber \\
=&
\alpha\displaystyle\sum_{j=0}^s\binom{j+n-1}{n-1}\binom{s-j+m-1}{m-1}(s-j),
\Label{constraint2'} \\
& \Big[\gamma(s)A(s,c_1(s),c_2(s))+B(s,c_1(s),c_2(s)) \Big]\times\binom{s+m+n-1}{m+n-1}^{-1}
=\alpha
\Label{constraint'}
\end{align}
hold for any $s\in\Z_{\ge 0}$ with the functions $A$ and $B$ defined by
\begin{eqnarray*}
&&A(s,u,v)\\
&:=&\sharp\left\{(k,l)\in\Z_{\ge 0}^{m+n}\Big|X_1(k)=u~or~v, X_2(k,l)=s\right\}\\
 &=&\binom{u+m-1}{m-1}\binom{s-u+n-1}{n-1}(1-\delta_{u,v})+\binom{v+m-1}{m-1}\binom{s-v+n-1}{n-1},\\
&&B(s,u,v)\\
&:=&\sharp\left\{(k,l)\in\Z_{\ge 0}^{m+n}\Big|X_1(k)<u~or~X_1(k)>v, X_2(k,l)=s\right\}\\
 &=&\displaystyle\sum_{a=0}^{u-1}\binom{a+m-1}{m-1}\binom{s-a+n-1}{n-1} +\displaystyle\sum_{b=0}^{v-1}\binom{b+m-1}{m-1}\binom{s-b+n-1}{n-1}.
\end{eqnarray*}
In the above test function $\ph$,
the conditions $(\ref{constraint'})$ and $(\ref{constraint2'})$ 
correspond to  the conditions $(\ref{cla.const})$ and $(\ref{cla.const2})$ in Theorem \ref{cla2}. Therefore, we get the following theorem.

\begin{thm}\Label{F test}
For the hypothesis testing problem $(\mathrm{H}\textrm{-}\mathrm{F})$, the test 
\[
T_{\alpha}^{[\mathrm{F}],m,n}=\displaystyle\sum \ph(k,l)|k,l\rangle\langle k,l|.
\]
 is a UMP unbiased test with level $\alpha$.
\end{thm}

We constructed a test $T_{\alpha}^{[\mathrm{F}],m,n}$ in the hypothesis testing problem $(\mathrm{H}\textrm{-}\mathrm{F})$. In the classical system, the hypothesis testing problem $(\mathrm{H}\textrm{-}\mathrm{F})$ corresponds to
\[
H_0:M=N~vs.~H_1:M\ne N~\hbox{with}~\{\mathrm{G}_{0,M}^{(m)}\mathrm{G}_{0,N}^{(n)}\}_{M,N\in\R_{\ge 0}}.
\]
 A $F$ test is a UMP unbiased test for the above classical hypothesis testing problem, and the test $T_{\alpha}^{[\mathrm{F}],m,n}$ is a UMP unbiased test for $(\mathrm{H}\textrm{-}\mathrm{F})$. Therefore, the test $T_{\alpha}^{[\mathrm{F}],m,n}$ can be regarded as a quantum counterpart of the $F$ test, and we call the test $T_{\alpha}^{[\mathrm{F}],m,n}$ a \textbf{quantum $F$ test} with $(m,n)$ degrees of freedom and with level $\alpha$.

%%%%%%%%%%%%%%%%%%%%%%%%%%%%%%%%%%%%%%%%%%%%%%%%%%%%%%%%%%%%%%%%
%%%%%%%%%%%%%%%%%%%%%%%%%%%%%%%%%%%%%%%%%%%%%%%%%%%%%%%%%%%%%%%%

\section{Hypothesis testing of the mean parameter}

We consider the hypothesis testing problem about the mean parameter for the quantum Gaussian states in this section, and derive a quantum counterpart of a $t$ test in subsection 5.2. 

%%%%%%%%%%%%%%%%%%%%%%%%%%%%%%%%%%%%%%%%%%%%%%%%%%%%%%%%%%%%%%%%%%%%%
%%%%%%%%%%%%%%%%%%%%%%%%%%%%%%%%%%%%%%%%%%%%%%%%%%%%%%%%%%%%%%%%%%%%%
\subsection{The case that the number parameter $N$ is known}

~

We consider the hypothesis testing problem
\[
H_0:|\theta|\in[0,R_0]~\text{vs.}~H_1:|\theta|\in(R_0,\infty)
\hbox{ with } 
\{\rho_{\theta,N}^{\otimes n}\}_{\theta\in\C}
 \eqno{(\mathrm{H}\textrm{-}1)}
\]
for $1\le n\in\N$ and $R_0\in\R_{\ge 0}$ when the number parameter $N$ is fixed. That is, we suppose that the number parameter $N$ is known. Then, the disturbance parameter space is $S^1=\{a\in\C||a|=1\}$ that represents the phase of mean parameter $\theta$. The UMP test for $(\mathrm{H}\textrm{-}1)$ does not exist as is shown in Appendix. Thus, our purpose is to derive a UMP min-max test in this subsection. 

When we change the parameterization as $\theta=r\mathrm{e}^{it}$, $(\mathrm{H}\textrm{-}1)$ is rewritten as 
\[
H_0:r\in[0,R_0]~\text{vs.}~H_1:r\in(R_0,\infty)
\hbox{ with } 
\{U_n^{\ast}(\rho_{\sqrt{n}r\mathrm{e}^{it},N}\otimes\rho_{0,N}^{\otimes (n-1)})U_n\}_{r\in\R_{\ge 0}, \mathrm{e}^{it}\in S^1}.
\]
We apply Theorem $\ref{thm3}$ to $(\mathrm{H}\textrm{-}1)$ in the following way:
%In Theorem $\ref{thm3}$, $(\mathrm{H}\textrm{-}1)$ is regarded as follows:
\begin{eqnarray}
&\mathcal{S}_1:=\{\rho_{\sqrt{n}r\mathrm{e}^{it},N}\otimes\rho_{0,N}^{\otimes (n-1)}\}_{r\in\R_{\ge 0}, \mathrm{e}^{it}\in S^1},~\mathcal{S}_2:=\{\rho_{0,N}^{\otimes (n-1)}\},~\mathcal{S}_3:=\phi,&\nonumber\\
&\Theta:=\R_{\ge 0},~ \Xi_1:=S^1,~ \Xi_2:=\{N\},~ U':=U_n,&\nonumber
\end{eqnarray}
where $\phi$ is the empty set, $\Xi_2$ has only one element $N$ and $U_n$ is the concentrating operator satisfying $(\ref{concentrating})$. Hence, a UMP min-max test for $(\mathrm{H}\textrm{-}1)$ is given as $U_n^*(T'\otimes I^{\otimes(n-1)})U_n$ by using a UMP min-max test $T'$ with level $\alpha$ for 
\[
H_0:r\in[0,\sqrt{n}R_0]~\text{vs.}~H_1:r\in(\sqrt{n}R_0,\infty)
\hbox{ with } 
\{\rho_{r\mathrm{e}^{it},N}\}_{r\in\R_{\ge 0}, \mathrm{e}^{it}\in S^1}.
\] 

Therefore, we only have to treat the hypothesis testing problem
\[
H_0:r\in[0,R_0]~\text{vs.}~H_1:r\in(R_0,\infty)
\hbox{ with } 
\{\rho_{r\mathrm{e}^{it},N}\}_{r\in\R_{\ge 0}, \mathrm{e}^{it}\in S^1}.
\eqno{(\mathrm{H}\textrm{-}\mathrm{A})}
\]
Here, we define a test $T_{R,N}^{\alpha}$ for $R\in\R_{\ge 0}$ by 
\[
T_{R,N}^{\alpha}:=\gamma_{R}|k_{R}\rangle\langle k_{R}|+\displaystyle\sum_{k=k_{R}+1}^{\infty}|k\rangle\langle k|,
\]
where $k_{R}\in\Z_{\ge 0}$ and $0<\gamma_{R}\le1$ is determined by level $\alpha$ as
\begin{eqnarray}
&1-\displaystyle\sum_{k=0}^{k_{R}}P_{R,N}(k) < \alpha \le 1-\displaystyle\sum_{k=0}^{k_{R}-1}P_{R,N}(k),& \Label{k}\\
&\gamma_{R}:=\frac{\alpha-\left(1-\sum_{k=0}^{k_{R}}P_{R,N}(k)\right)}{P_{R,N}(k_{R})}, \Label{gamma}&
\end{eqnarray}
where $P_{R,N}$ is the probability distribution in Lemma \ref{prob}. The test $T_{R,N}^{\alpha}$ is guaranteed to be with level $\alpha$ by $(\ref{k})$ and $(\ref{gamma})$ as is shown in the proof of the next theorem.
%,  and is a UMP min-max test for $(\mathrm{H}\textrm{-}\mathrm{A})$ as is shown in the next theorem.

\begin{thm}\Label{thmH-A}
For the hypothesis $(\mathrm{H}\textrm{-}\mathrm{A})$, the test $T_{R_0,N}^{\alpha}$ is a UMP min-max test with level $\alpha$.
\end{thm}

Therefore, we get the following proposition.

\begin{prp}\Label{mean.min-max}
For the hypothesis $(\mathrm{H}\textrm{-}1)$, the test 
\[
T_{\alpha,R_0,N}^{[1],n}:=U_n^{\ast}(T_{\sqrt{n}R_0,N}^{\alpha}\otimes I^{\otimes (n-1)})U_n
\]
 is a UMP min-max test with level $\alpha$.
\end{prp}

\hspace{-1em}\textit{Proof of Theorem \ref{thmH-A}:} 
The representation $\{S_{e^{it}}\}_{e^{it}\in S^1}$ of $S^1=\{a\in\C||a|=1\}$ on $L^2(\R)$ satisfies 
\[
S_{e^{it}} \rho_{r\mathrm{e}^{is}} S_{e^{it}}^*=\rho_{r\mathrm{e}^{i(s+t)}},
\]
and is covariant concerning the disturbance parameter space $S^1$ in the sense of the subsection $3.2$. Due to Theorem \ref{cpt.H-S}, it is enough to prove that $T_{R_0,N}^{\alpha}$ is a UMP $S^1$-invariant test concerning the representation $\{S_{e^{it}}\}_{e^{it}\in S^1}$ with level $\alpha$. 

The type I error probability $\alpha_{T_{R_0,N}^{\alpha}}(r\mathrm{e}^{it})=\mathrm{Tr}\rho_{{r}\mathrm{e}^{it},  N}T_{R_0,N}^{\alpha}$ is monotonically increasing with respect to $r$ and the equation $\mathrm{Tr}\rho_{{R_0}\mathrm{e}^{it},  N}T_{R_0,N}^{\alpha}=\alpha$ holds by the conditions $(\ref{k})$ and $(\ref{gamma})$. Hence, $T_{R_0,N}^{\alpha}$ is a test with level $\alpha$. In addition, $T_{R_0,N}^{\alpha}$ is $S^1$-invariant test due to Lemma \ref{inv2}.

Let $T=\displaystyle\sum_{k=0}^{\infty}t_k |k\rangle\langle k|$ be an arbitrary $S^1$ invariant test concerning the representation $\{S_{e^{it}}\}_{e^{it}\in S^1}$. Since the Laguerre polynomials satisfy 
\begin{align}\Label{L.polynomial}
L_{k+1}(x)=\frac{1}{k+1}((-x+2k+1)L_k(x)-kL_{k-1}(x)),
\end{align}
it is easily shown that $L_{k}(-x)/L_{l}(-x)$~$(k>l)$ and $L_{k}(-x)/L_{l}(-x)$~$(k<l)$ are monotone increasing and decreasing with respect to $x\in\R_{> 0}$ respectively by applying the inductive method to (\ref{L.polynomial}).

Then
\[
\begin{array}{rll}
\mathrm{Tr}\rho_{r\mathrm{e}^{it}, N}T_{R_0,N}^{\alpha}
&=&\gamma_{R_0}\frac{1}{N+1}\left(\frac{N}{N+1}\right)^{k_{R_0}}\mathrm{e}^{-\frac{r^2}{N(N+1)}}L_{k_{R_0}}\left(-\frac{r^2}{N(N+1)}\right)\\
&&~~~+\displaystyle\sum_{k=k_{R_0}+1}^{\infty}\frac{1}{N+1}\left(\frac{N}{N+1}\right)^k\mathrm{e}^{-\frac{r^2}{N(N+1)}}L_k\left(-\frac{|\theta|^2}{N(N+1)}\right),\\
\mathrm{Tr}\rho_{r\mathrm{e}^{it}, N}T
&=&\displaystyle\sum_{k=0}^{\infty}t_k\frac{1}{N+1}\left(\frac{N}{N+1}\right)^k\mathrm{e}^{-\frac{r^2}{N(N+1)}}L_k\left(-\frac{r^2}{N(N+1)}\right),
\\
\end{array}
\]
which imply
\[
\begin{array}{ll}
&\left({\mathrm{Tr}(\rho_{r\mathrm{e}^{it},  N}T_{R_0,N}^{\alpha})-\mathrm{Tr}(\rho_{r\mathrm{e}^{it}, 
N}T)}\right)\Big/{\frac{1}{N+1}\left(\frac{N}{N+1}\right)^{k_{R_0}}\mathrm{e}^{-\frac{r^2}{N(N+1)}}L_{k_{R_0}}\left(-\frac{r^2}{N(N+1)}\right)}\\
=&-\displaystyle\sum_{k=0}^{k_{R_0}-1}t_{k}\left(\frac{N}{N+1}\right)^{k-k_{R_0}}\frac{L_k\left(-\frac{r^2}{N(N+1)}\right)}{L_{k_{R_0}}\left(-\frac{r^2}{N(N+1)}\right)}
+(\gamma_{R_0}t_{k_{R_0}})\\
&+\displaystyle\sum_{k=k_{R_0}+1}^{k_{R_0}-1}(1-t_{k})\left(\frac{N}{N+1}\right)^{k-k_{R_0}}\frac{L_k\left(-\frac{r^2}{N(N+1)}\right)}{L_{k_{R_0}}\left(-\frac{r^2}{N(N+1)}\right)}.
\end{array}
\]
Thus, the above function is monotonically increasing with respect to $r$, and hence, the following holds in $r>R_0$.
\[
\begin{array}{ll}
&\left({\mathrm{Tr}(\rho_{r\mathrm{e}^{it}, N}T_{R_0,N}^{\alpha})-\mathrm{Tr}(\rho_{r\mathrm{e}^{it},
N}T)}\right)\Big/{\frac{1}{N+1}\left(\frac{N}{N+1}\right)^{k_{R_0}}\mathrm{e}^{-\frac{r^2}{N(N+1)}}L_{k_{R_0}}\left(-\frac{r^2}{N(N+1)}\right)}\\
&\ge(\alpha-\mathrm{Tr}(\rho_{{R_0},
N}T))\Big/{\frac{1}{N+1}\left(\frac{N}{N+1}\right)^{k_{R_0}}\mathrm{e}^{-\frac{{R_0}^2}{N(N+1)}}L_{k_{R_0}}\left(-\frac{{R_0}^2}{N(N+1)}\right)}\\
&\ge 0.
\end{array}
\]
Therefore we get
\[
\beta_{T_{R_0,N}^{\alpha}}(\rho_{r\mathrm{e}^{it},N})
=1-\mathrm{Tr}(\rho_{r\mathrm{e}^{it}, N}T_{R_0,N}^{\alpha})
\le 1-\mathrm{Tr}(\rho_{r\mathrm{e}^{it}, N}T)
=\beta_{T}(\rho_{r\mathrm{e}^{it},N})
\]
for $r>R_0$.

\endproof
%\hspace{32.5em}■

In particular, since we can regard a quantum Gaussian state $\rho_{\theta,0}$ at the number parameter $N=0$ as a coherent state $|\theta)(\theta|$, the test $T_{\alpha,R_0,0}^{[1],n}$ gives a UMP min-max test for the hypothesis testing problem on coherent states:
\[
H_0:|\theta|\in[0,R_0]~\text{vs.}~H_1:|\theta|\in(R_0,\infty)
\hbox{ with } \{|\theta)(\theta|^{\otimes n}\}_{\theta\in\C}
\]
for $1\le n\in\N$.

%%%%%%%%%%%%%%%%%%%%%%%%%%%%%%%%%%%%%%%%%%%%%%%%%%%%%%%%%%%%%%%%%%%%%
%%%%%%%%%%%%%%%%%%%%%%%%%%%%%%%%%%%%%%%%%%%%%%%%%%%%%%%%%%%%%%%%%%%%%
\subsection{The case that the number parameter is unknown: $t$ test}

~

In this subsection, we propose a quantum counter part of a $t$ test. We consider 
the hypothesis testing problem
\[
H_0:|\theta|\in[0,R_0]~\text{vs.}~H_1:|\theta|\in(R_0,\infty)
\hbox{ with } \{\rho_{\theta,N}^{\otimes n}\}_{\theta\in\C,N\in\R_{>0}} \eqno{(\mathrm{H}\textrm{-}2)}
\]
for $2\le n\in\N$ and $R_0\in\R_{\ge0}$. But it is difficult to derive an optimal test for the above hypothesis testing problem for an arbitrary $R_0\in\R_{\ge 0}$ since the number parameter $N$ is unknown. Hence, we consider the hypothesis testing problem $(\mathrm{H}\textrm{-}2)$ at $R_0=0$:
\[
H_0:|\theta|\in\{0\}~\text{vs.}~H_1:|\theta|\in(0,\infty)
\hbox{ with } \{\rho_{\theta,N}^{\otimes n}\}_{\theta\in\C,N\in\R_{>0}}.
\]
When we change the parameterization as $\theta=r\mathrm{e}^{it}$, $(\mathrm{H}\textrm{-}2)$ at $R_0=0$ is rewritten as 
\begin{eqnarray}
&H_0:r\in\{0\}~\text{vs.}~H_1:r\in(0,\infty)&\nonumber\\
&\hbox{ with } 
\{U_n^{\ast}(\rho_{\sqrt{n}r\mathrm{e}^{it},N}\otimes\rho_{0,N}^{\otimes (n-1)})U_n\}_{r\in\R_{\ge 0}, \mathrm{e}^{it}\in S^1,N\in\R_{>0}}.&\nonumber
\end{eqnarray}
We apply Lemma $\ref{lem1}$ to $(\mathrm{H}\textrm{-}2)$ in the following way:
%. In Lemma $\ref{lem1}$, $(\mathrm{H}\textrm{-}2)$ is regarded as follows:
\begin{eqnarray}
&\mathcal{S}:=\{\rho_{\sqrt{n}r\mathrm{e}^{it},N}\otimes\rho_{0,N}^{\otimes (n-1)}\}_{r\in\R_{\ge 0}, \mathrm{e}^{it}\in S^1,N\in\R_{>0}},&\nonumber\\
&\Theta:=\R_{\ge 0}(\ni r),~ \Xi_1:=S^1\times \R_{\ge 0}(\ni(\mathrm{e}^{it},N)),~ U:=U_n,&\nonumber
\end{eqnarray}
where $\phi$ is the empty set and $U_n$ is the concentrating operator satisfying $(\ref{concentrating})$. Hence, a UMP unbiased min-max test for $(\mathrm{H}\textrm{-}2)$ at $R_0=0$ is given as $U_n^*T'U_n$ by using a UMP unbiased min-max test $T'$ with level $\alpha$ for 
\[
H_0:r\in\{0\}~\text{vs.}~H_1:r\in(0,\infty)
\hbox{ with } 
\{\rho_{\sqrt{n}r\mathrm{e}^{it},N}\otimes\rho_{0,N}^{\otimes (n-1)}\}_{r\in\R_{\ge 0}, \mathrm{e}^{it}\in S^1}.
\]

Therefore, we only have to treat the hypothesis testing problem
\[
H_0:r\in\{0\}~\text{vs.}~H_1:r\in(0,\infty)
\hbox{ with } \{\rho_{r\mathrm{e}^{it},N}\otimes\rho_{0,N}^{\otimes n}\}_{r\in\R_{\ge 0}, \mathrm{e}^{it}\in S^1,N\in\R_{>0}}. \eqno{(\mathrm{H}\textrm{-}\mathrm{t})}
\]
Here, the disturbance parameter space is $S^1\times\R_{>0}$. We define a test $T_{\alpha}^{[\mathrm{t}],n}$ by
\[
T_{\alpha}^{[\mathrm{t}],n}=\displaystyle\sum_{k=(k_0,\cdot\cdot\cdot,k_n)\in\Z_{\ge 0}^{n+1}}\ph'^{(n)}({k})|k\rangle\langle k|
\]
where the test function
\[
\ph'^{(n)}(k):=
\left\{
\begin{array}{ll}
0&(k_0<c(s(k)))\\
\gamma(s(k))&(k_0=c(s(k)))\\
1&(k_0>c(s(k)))
\end{array}
\right.
\]
depends on the total counts $s(k):=\sum_{j=0}^n k_j$ of $k=(k_0,\cdot\cdot\cdot,k_n)$ and functions
\begin{eqnarray}\Label{c,gamma}
c:\Z_{\ge 0}\to\Z_{\ge 0} ,~\gamma:\Z_{\ge 0}\to (0,1]&
\end{eqnarray}
are determined as
\begin{eqnarray}
&\displaystyle\sum_{l=c(s)+1}^s\binom{s-l+n-1}{n-1}
<\alpha \binom{s+n}{n}
\le \displaystyle\sum_{l=c(s)}^s\binom{s-l+n-1}{n-1}& \Label{c(s)'},\\
&\gamma(s):=\frac{\alpha\binom{s+n}{n}-\sum_{l=c(s)+1}^s\binom{s-l+n-1}{n-1}}{\binom{s-c(s)+n-1}{n-1}}& \Label{gamma(s)'}
\end{eqnarray}
for each total counts $s\in\Z_{\ge 0}$. The test $T_{\alpha}^{[\mathrm{t}],n}$ is guaranteed to be with level $\alpha$ by $(\ref{c,gamma})$ as is shown in the proof of the next theorem, and a UMP unbiased min-max test for $(\mathrm{H}\textrm{-}\mathrm{t})$ as is shown in the next theorem.

\begin{thm}\Label{t-test}
For the hypothesis $(\mathrm{H}\textrm{-}\mathrm{t})$, the test $T_{\alpha}^{[\mathrm{t}],n}$ is a UMP unbiased min-max test with level $\alpha$.
\end{thm}

Therefore, we get the following proposition.

\begin{prp}\Label{mean.min-max2}
For the hypothesis $(\mathrm{H}\textrm{-}2)$ at $R_0=0$, the test 
\[
T^{[2],n}_{\alpha}:=U_n^{\ast}T_{\alpha}^{[\mathrm{t}],n-1}U_n
\]
 is a UMP unbiased min-max test with level $\alpha$.
\end{prp}

\hspace{-1.5em}\textit{Proof of Theorem} \ref{t-test}\textit{:}
Due to Theorem \ref{cpt.H-S}, it is enough to prove that $T_{\alpha}^{[\mathrm{t}],n}$ is a UMP unbiased $S^1$-invariant test concerning the representation $\{S_{e^{it}}\otimes I^{\otimes n}\}_{e^{it}\in S^1}$ with level $\alpha$. 

Here, a test $T$ is an unbiased test with level $\alpha$ if and only if 
\[
\begin{array}{ll}
&\mathrm{Tr}(\rho_{0,\mathrm{e}^{it},N}^{\otimes n}T)\\
 =&\displaystyle\sum_{s=0}^{\infty}\left(\displaystyle\sum_{k\in\Z^{n+1}: s(k)=s}\langle k|T|k\rangle{\binom{s+n-1}{n-1}}^{-1}\right){\binom{s+n-1}{n-1}}\left(\frac{1}{N+1}\right)^n\left(\frac{N}{N+1}\right)^s\\
 =&\alpha
\end{array}
\]
holds for any $N\in\R_{>0}$. The above condition is equivalent to 
\begin{align}\Label{alphaconst'}
\displaystyle\sum_{k\in\Z^{n+1}: s(k)=s}\langle k|T|k\rangle{\binom{s+n-1}{n-1}}^{-1}=\alpha
\end{align}
for any $s\in\Z_{\ge 0}$. By the conditions $(\ref{c(s)'})$ and $(\ref{gamma(s)'})$,
\begin{align}\Label{alphaconst2'}
\displaystyle\sum_{k\in\Z^{n+1}: s(k)=s}\ph'^{(n)}({k}){\binom{s+n-1}{n-1}}^{-1}=\alpha
\end{align}
holds, and hence $T_{\alpha}^{[\mathrm{t}],n}$ is an unbiased test with level $\alpha$. In addition, due to Lemma $\ref{inv2}$, $T_{\alpha}^{[\mathrm{t}],n}$ is $S^1$-invariant test  concerning the representation $\{S_{e^{it}}\otimes I^{\otimes n}\}_{e^{it}\in S^1}$.

Let $T$ be an arbitrary unbiased $S^1$-invariant test  concerning the representation $\{S_{e^{it}}\otimes I^{\otimes n}\}_{e^{it}\in S^1}$. Then ${T}$ is represented as 
\[
{T}=\displaystyle\sum_{k_0=0}^{\infty}|k_0\rangle\langle k_0|\otimes T_{k_0}
%\displaystyle\sum_{k_0,\cdot\cdot\cdot,k_n=0}^{\infty}t_{k_0,\cdot\cdot\cdot,k_n}|k_0,\cdot\cdot\cdot,k_n\rangle\langle k_0,\cdot\cdot\cdot,k_n|  
\]
by Lemma $\ref{inv2}$. %and the commutativity of quantum Gaussian states $\{\rho_{0,N}^{\otimes n}\}_{N\in\R_{>0}}$. 

Since
\[
\begin{array}{lll}
&&\beta_{T}(\rho_{r\mathrm{e}^{it},N}\otimes\rho_{0,N}^{\otimes n})-\beta_{T_{\alpha}^{[\mathrm{t}],n}}(\rho_{r\mathrm{e}^{it},N}\otimes\rho_{0,N}^{\otimes n})\\
&=& \mathrm{Tr}((\rho_{r\mathrm{e}^{it},N}\otimes\rho_{0,N}^{\otimes n})T_{\alpha}^{[\mathrm{t}],n})-\mathrm{Tr}((\rho_{r\mathrm{e}^{it},N}\otimes\rho_{0,N}^{\otimes n})T)\\
 &=&\displaystyle\sum_{s=0}^{\infty}\displaystyle\sum_{k\in\Z_{\ge 0}^{n+1}:s(k)=s}(\ph'^{(n)}({k})-t_{k})
 L_{c(s)}\left(\frac{-r^2}{N(N+1)}\right)\left(\frac{1}{N+1}\right)^n\left(\frac{N}{N+1}\right)^s e^{\frac{-r^2}{N+1}}\\
 &=&\displaystyle\sum_{s=0}^{\infty}\left\{\displaystyle\sum_{k\in\Z_{\ge 0}^{n+1}:s(k)=s}(\ph'^{(n)}({k})-t_{k})\frac{L_{k_0}\left(\frac{-r^2}{N(N+1)}\right)}{L_{c(s)}\left(\frac{-r^2}{N(N+1)}\right)}\right\}\\
 &&\times L_{c(s)}\left(\frac{-r^2}{N(N+1)}\right)\left(\frac{1}{N+1}\right)^n\left(\frac{N}{N+1}\right)^s e^{\frac{-r^2}{N+1}}
\end{array}
\]
where $t_k:=\langle k_1,\cdot\cdot\cdot,k_n|T_{k_0}|k_1,\cdot\cdot\cdot,k_n\rangle$ for $k=(k_0,k_1,\cdot\cdot\cdot,k_n)\in\Z_{\ge 0}^{n+1}$, 
it is enough to show 
\begin{equation}\Label{inequality}
\displaystyle\sum_{k\in\Z_{\ge 0}^{n+1}:s(k)=s}(\ph'^{(n)}({k})-t_{k})\frac{L_{k_0}\left(-x\right)}{L_{c(s)}\left(-x\right)}\ge 0
\end{equation}
for any $s\in\Z_{\ge 0}$ and $x\in\R_{> 0}$ in order to prove that $\beta_T-\beta_{T_{\alpha}^{[\mathrm{t}],n}}$ is nonnegative. 
\[
\begin{array}{ll}
&\displaystyle\sum_{k\in\Z_{\ge 0}^{n+1}:s(k)=s}(\ph'^{(n)}({k})-t_{k})\frac{L_{k_0}\left(-x\right)}{L_{c(s)}\left(-x\right)}\\
 =&\displaystyle\sum_{k_0=0}^{s}\displaystyle\sum_{(k_1,\cdot\cdot\cdot,k_n)\in\Z_{\ge 0}^{n}:s(k)=s}(\ph'^{(n)}({k})-t_{k})\frac{L_{k_0}\left(-x\right)}{L_{c(s)}\left(-x\right)}\\
 =&\displaystyle\sum_{ k_0=0}^{c(s)-1}\displaystyle\sum_{(k_1,\cdot\cdot\cdot,k_n)\in\Z_{\ge 0}^{n}:s(k)=s}(-t_{k})\frac{L_{k_0}\left(-x\right)}{L_{c(s)}\left(-x\right)}
 +\displaystyle\sum_{k_0=c(s)}\displaystyle\sum_{(k_1,\cdot\cdot\cdot,k_n)\in\Z_{\ge 0}^{n}:s(k)=s}(\gamma(s)-t_{k})\\
 &+\displaystyle\sum_{k_0=c(s)+1}^{\infty}\displaystyle\sum_{(k_1,\cdot\cdot\cdot,k_n)\in\Z_{\ge 0}^{n}:s(k)=s}(1-t_{k})\frac{L_{k_0}\left(-x\right)}{L_{c(s)}\left(-x\right)}.
\end{array}
\]
The above function equals 0 at $x=0$ by (\ref{alphaconst'}) and (\ref{alphaconst2'}), and increases with respect to $x\in\R_{> 0}$ in the same way as the proof of Theorem \ref{mean.min-max}. Thus, the inequality $(\ref{inequality})$ is proved.

\endproof
%\hspace{32.5em}■

We constructed a test $T_{\alpha}^{[\mathrm{t}],n}$ in the hypothesis testing problem $(\mathrm{H}\textrm{-}\mathrm{t})$. In the classical system, the hypothesis testing problem $(\mathrm{H}\textrm{-}\mathrm{t})$ corresponds to
\[
H_0:|\theta|\in\{0\}~vs.~H_1:|\theta|\in(0,\infty)~\hbox{with}~ \{G_{\theta,N}^{(n)}\}_{ \theta\in\R, N\in\R_{>0}},
\]
where $G_{\theta,N}$ is a Gaussian distribution. 
A $t$ test is a UMP min-max test for the above hypothesis testing problem. In addition, since the type I and II error probabilities on the $t$ test does not depend on the sign of the mean parameter of a Gaussian distribution. Similarly, the type I and II error probabilities on the test $T_{\alpha}^{[\mathrm{t}],n}$ does not depend on the phase of the mean parameter of a quantum Gaussian state. Therefore, the test $T_{\alpha}^{[\mathrm{t}],n}$ can be regarded as a quantum counterpart of the $t$ test. Hence, we call the test $T_{\alpha}^{[\mathrm{t}],n}$ in the above theorem a \textbf{quantum} \mbox{\boldmath ${ t}$} \textbf{test} with $n$ degrees of freedom and level $\alpha$. 

When the hypothesis testing problem $(\mathrm{H}\textrm{-}2)$ at $R_0\ne 0$, it is not solved whether there exists an optimal test for the hypothesis testing problem $(\mathrm{H}\textrm{-}2)$. This quantum hypothesis testing problem is analogous to the hypothesis testing problem for the size of the mean parameter of the Gaussian distribution with unknown variance in classical hypothesis testing. The problem is called the bioequivalence problem \cite{BHM} and is not solved whether there exists an optimal test for it. But the problem appears in several situations including medicine and pharmacy, and is expected to be solved from the demand of not only the theoretical aspect but also the application aspect. The problem $(\mathrm{H}\textrm{-}2)$ is expected to be also important in the quantum hypothesis testing because of the importance of the bioequivalence problem.

%%%%%%%%%%%%%%%%%%%%%%%%%%%%%%%%%%%%%%%%%%%%%%%%%%%%
%%%%%%%%%%%%%%%%%%%%%%%%%%%%%%%%%%%%%%%%%%%%%%%%%%%%
\section{Hypothesis testing of the mean parameters for two kinds of quantum Gaussian states}

We consider the hypothesis testing about the consistency of the mean parameters for two kinds of quantum Gaussian states $\{\rho_{\theta,N}^{\otimes m}\otimes \rho_{\eta,N}^{\otimes n}\}$.

%%%%%%%%%%%%%%%%%%%%%%%%%%%%%%%%%%%%%%%%%%%%%%%%%%%%%%
%%%%%%%%%%%%%%%%%%%%%%%%%%%%%%%%%%%%%%%%%%%%%%%%%%%%%%
\subsection{The case that the number parameter is known}

~

We consider the hypothesis testing problem
\[
H_0:\theta=\eta~vs.~H_1:\theta\ne\eta~\hbox{with}~\{\rho_{\theta,N}^{\otimes m}\otimes \rho_{\eta,N}^{\otimes n}\}_{\theta,\eta\in\C}
\eqno{(\mathrm{H}\textrm{-}3)}
\]
for $1\le n\in\N$ when the number parameter $N$ is fixed. That is, we suppose that the number parameter $N$ is known. 

When we change the parameterization as $m\theta+n\eta=a, \theta-\eta=r\mathrm{e}^{it}$, $(\mathrm{H}\textrm{-}3)$ is rewritten as 
\begin{eqnarray}
&H_0:r\in\{0\}~vs.~H_1:r\in(0,\infty)&\nonumber\\
&~\hbox{with}~
\{U'^*_{m,n}(\rho_{a,N}\otimes\rho_{r\mathrm{e}^{it},N}\otimes\rho_{0,N}^{\otimes (m+n-2)})U'_{m,n}\}_{a\in\C,r\in\R_{\ge 0},\mathrm{e}^{it}\in S^1}.&\nonumber
\end{eqnarray}
We apply Theorem $\ref{thm3}$ to $(\mathrm{H}\textrm{-}3)$ in the following way:
%. In Theorem $\ref{thm3}$, $(\mathrm{H}\textrm{-}3)$ is regarded as follows:
\begin{eqnarray}
&\mathcal{S}_1:=\{\rho_{r\mathrm{e}^{it},N}\}_{r\in\R_{\ge 0}, \mathrm{e}^{it}\in S^1},~ \mathcal{S}_2:=\{\rho_{a,N}\otimes\rho_{0,N}^{\otimes (m+n-2)}\}_{a\in\C}, \mathcal{S}_3:=\phi,&\nonumber\\
&\Theta:=\R_{\ge0},~ \Xi_1:=S^1,~ \Xi_2:=\C,~ U':=U'_{m,n},&\nonumber
\end{eqnarray}
where $\phi$ is the empty set, and $U'_{m,n}$ is the unitary operator satisfying $(\ref{U2})$. Hence, a UMP min-max test for $(\mathrm{H}\textrm{-}3)$ is given as $U'^*_{m,n}(I\otimes T'\otimes I^{\otimes(m+n-2)})U'_{m,n}$ by using a UMP min-max test $T'$ with level $\alpha$ for 
\begin{eqnarray}
&H_0:r\in\{0\}~vs.~H_1:r\in(0,\infty)
~\hbox{with}~\{\rho_{r\mathrm{e}^{it},N}\}_{r\in\R_{\ge 0},\mathrm{e}^{it}\in S^1}.& \nonumber
\end{eqnarray}
Since this hypothesis testing problem is nothing but $(\mathrm{H}\textrm{-}\mathrm{A})$, 
Theorem \ref{thmH-A} guarantees that 
$T_{R_0,N}^{\alpha}$ is a UMP min-max test for the above testing problem.
Therefore, we get the following proposition.

\begin{prp}
For the hypothesis testing problem $(\mathrm{H}\textrm{-}3)$, the test 
\[
T_{\alpha,N}^{[3],m,n}:={U'}_{m,n}^{\ast}(I\otimes T_{0,N}^{1,\alpha}\otimes I^{\otimes (m+n-2)}){U'}_{m,n}
\]
 is a UMP min-max test with level $\alpha$.
\end{prp}

In particular, the test $T_{\alpha,0}^{[3],m,n}$ gives a UMP min-max test for the hypothesis testing problem on coherent states:
\[
H_0:\theta=\eta~vs.~H_1:\theta\ne\eta~\hbox{with}~\{|\theta)(\theta|^{\otimes m}\otimes |\eta)(\eta|^{\otimes n}\}_{\theta,\eta\in\C}
\]
for $2\le n\in\N$.

%%%%%%%%%%%%%%%%%%%%%%%%%%%%%%%%%%%%%%%%%%%%%%%%%%%%%%%%%%%%%%%%%%%%%
%%%%%%%%%%%%%%%%%%%%%%%%%%%%%%%%%%%%%%%%%%%%%%%%%%%%%%%%%%%%%%%%%%%%%
\subsection{The case that the number parameter is unknown: $t$ test}

~

We consider the hypothesis testing problem
\[
H_0:\theta=\eta~vs.~H_1:\theta\ne\eta
~\hbox{with}~
\{\rho_{\theta,N}^{\otimes m}\otimes \rho_{\eta,N}^{\otimes n}\}_{\theta,\eta\in\C,N\in\R_{>0}}
\eqno{(\mathrm{H}\textrm{-}4)} 
\]
for $2\le n\in\N$. 

When we change the parameterization as $m\theta+n\eta=a, \theta-\eta=r\mathrm{e}^{it}$, $(\mathrm{H}\textrm{-}4)$ is rewritten as
\begin{eqnarray}
&H_0:r\in\{0\}~vs.~H_1:r\in(0,\infty)&\nonumber\\
&~\hbox{with}~
\{U'^*_{m,n}(\rho_{a,N}\otimes\rho_{r\mathrm{e}^{it},N}\otimes\rho_{0,N}^{\otimes (m+n-2)})U'_{m,n}\}_{a\in\C,r\in\R_{\ge 0},\mathrm{e}^{it}\in S^1,N\in\R_{>0}}.\nonumber&
\end{eqnarray}
We apply Theorem $\ref{thm3}$ for $(\mathrm{H}\textrm{-}4)$ in the following way:
%. In Theorem $\ref{thm3}$, $(\mathrm{H}\textrm{-}4)$ is regarded as follows:
\begin{eqnarray}
&\mathcal{S}_1:=\{\rho_{r\mathrm{e}^{it},N}\otimes\rho_{0,N}^{\otimes (m+n-2)}\}_{r\in\R_{\ge 0}, \mathrm{e}^{it}\in S^1,N\in\R_{>0}},~ \mathcal{S}_2:=\phi, ~\mathcal{S}_3:=\{\rho_{a,N}\}_{a\in\C,N\in\R_{\ge 0}},&\nonumber\\
&\Theta:=\R_{\ge0}(\ni r),\Xi_1:=S^1\times\R_{\ge 0}(\ni(\mathrm{e}^{it},N)), \Xi_3=G:=\C,V:=W, U':=U'_{m,n},&\nonumber
\end{eqnarray}
where $\phi$ is the empty set, $U'_{m,n}$ is the unitary operator satisfying $(\ref{U2})$, and $W$ is the mean shift operators $\{W_{\theta}\}_{\theta\in\C}$ as a representation of $\C$ defined in $(\ref{mean shift operator})$.
Hence, a UMP unbiased min-max test for $(\mathrm{H}\textrm{-}4)$ is given as $U'^*_{m,n}(I\otimes T')U'_{m,n}$ by using a UMP unbiased min-max test $T'$ with level $\alpha$ for 
\begin{eqnarray}
&H_0:r\in\{0\}~vs.~H_1:r\in(0,\infty)
~\hbox{with}~\{\rho_{r\mathrm{e}^{it},N}\otimes\rho_{0,N}^{\otimes (m+n-2)}\}_{r\in\R_{\ge 0},\mathrm{e}^{it}\in S^1}.& \nonumber
\end{eqnarray}
Since this hypothesis testing problem is nothing but $(\mathrm{H}\textrm{-}\mathrm{t})$, 
Theorem \ref{t-test} guarantees that 
$T_{\alpha}^{[\mathrm{t}],{m+n-2}}$ is a UMP unbiased min-max test for the above testing problem.
Therefore, we get the following proposition.

\begin{prp}

For the hypothesis testing problem $(\mathrm{H}\textrm{-}4)$, the test 
\[
T_{\alpha,N}^{[4],m,n}:={U'}_{m,n}^{\ast}(I\otimes T_{\alpha}^{[\mathrm{t}],{m+n-2}}){U'}_{m,n}
\]
is a UMP unbiased min-max test with level $\alpha$.

\end{prp}

%%%%%%%%%%%%%%%%%%%%%%%%%%%%%%%%%%%%%%%%%%%%%%%%%%%%%%%%%%%%%%
%%%%%%%%%%%%%%%%%%%%%%%%%%%%%%%%%%%%%%%%%%%%%%%%%%%%%%%%%%%%%%
\section{Hypothesis testing of the number parameter}

We consider the hypothesis testing problem about the number parameter for quantum Gaussian states.

%%%%%%%%%%%%%%%%%%%%%%%%%%%%%%%%%%%%%%%%%%%%%%%%%%%%%%%%%%%%%%
%%%%%%%%%%%%%%%%%%%%%%%%%%%%%%%%%%%%%%%%%%%%%%%%%%%%%%%%%%%%%%
\subsection{The case that the mean parameter is known: $\chi^2$ test}

~

In this subsection, we treat the hypothesis testing problem
\[
H_0:N\in[0,N_0]~vs.~H_1:N\in (N_0,\infty)~\hbox{with}~\{\rho_{\theta,N}^{\otimes n}\}_{N\in\R_{>0}} \eqno{(\mathrm{H}\textrm{-}5)}
\]
for $1\le n\in\N$ and $N\in\R_{>0}$ when the mean parameter $\theta$ is fixed. That is, we suppose that the mean parameter $\theta$ is known. Then, the disturbance parameter space is empty. Note that an optimal test for the classical analogue of $(\mathrm{H}\textrm{-}5)$ is given by a $\chi^2$ test.

We apply Lemma $\ref{lem1}$ to $(\mathrm{H}\textrm{-}5)$ in the following way:
%. In Lemma $\ref{lem1}$, $(\mathrm{H}\textrm{-}5)$ is regarded as follows:
\[
\mathcal{S}:=\{\rho_{0,N}^{\otimes n}\}_{N\in\R_{\ge0}},~\Theta:=\R_{\ge0},~\Xi_1:=\phi,~U:=W_{-\theta}^{\otimes n},
\]
where $\phi$ is the empty set and $W_{\theta}$ is the mean shift operator defined in $(\ref{mean shift operator})$. Hence, a UMP test for $(\mathrm{H}\textrm{-}5)$ is given as $W_{-\theta}^{\otimes n *}T'W_{-\theta}^{\otimes n}$ by using a UMP test $T'$ with level $\alpha$ for 
\[
H_0:N\in (0,N_0]~vs.~H_1:N\in(N_0,\infty)
\hbox{ with }
\{\rho_{0,N}^{\otimes n}\}_{N\in\R_{> 0}} 
\]
Since this hypothesis testing problem is nothing but $(\mathrm{H}\textrm{-}\mathrm{\chi^2})$, 
Theorem \ref{chi2 test} guarantees that 
$T_{\alpha}^{[\mathrm{\chi^2}],n}$ is a UMP test for the above testing problem.
Therefore, we get the following proposition.

\begin{prp}\Label{numUMP}
For the hypothesis testing problem $(\mathrm{H}\textrm{-}5)$, the test 
\[
T_{\alpha,\theta, N_0}^{[5],n}:=(W_{-\theta}^{\otimes n})^{\ast} T_{\alpha}^{[\mathrm{\chi^2}],n} (W_{-\theta}^{\otimes n})
\]
 is a UMP test with level $\alpha$.
\end{prp}

In the above derivation,
we employ Theorem \ref{chi2 test} as well as our reduction method.
Since Theorem \ref{chi2 test} is shown by Theorem \ref{cla1},
application of the classical result (Theorem \ref{cla1}) is essential in the above derivation.
Similar observations can be applied to Propositions \ref{num.min-max}, \ref{2numUMP}, and \ref{2numUMP2}. 

%%%%%%%%%%%%%%%%%%%%%%%%%%%%%%%%%%%%%%%%%%%%%%%%%%%%%%%%%%%%%%%%%%%%%
%%%%%%%%%%%%%%%%%%%%%%%%%%%%%%%%%%%%%%%%%%%%%%%%%%%%%%%%%%%%%%%%%%%%%
\subsection{The case that the mean parameter is unknown: $\chi^2$ test}

~

In this subsection, we treat the hypothesis testing problem
\[
H_0:N\in(0,N_0]~vs.~H_1:N\in (N_0,\infty)~\hbox{with}~\{\rho_{\theta,N}^{\otimes n}\}_{\theta\in\C,N\in\R_{>0}}. \eqno{(\mathrm{H}\textrm{-}6)}
\]
for $2\le n\in\N$ and $N\in\R_{>0}$. Then, the disturbance parameter space is $\C$. Note that an optimal test for the classical analogue of $(\mathrm{H}\textrm{-}6)$ is given by a $\chi^2$ test.

We apply Theorem $\ref{thm3}$ to $(\mathrm{H}\textrm{-}6)$ in the following way:
%. In Theorem $\ref{thm3}$, $(\mathrm{H}\textrm{-}6)$ is regarded as follows:
\begin{eqnarray}
&\mathcal{S}_1:=\{\rho_{0,N}^{\otimes (n-1)}\}_{N\in\R_{\ge0}},~\mathcal{S}_2:=\phi,~\mathcal{S}_3:= \{\rho_{\sqrt{n}\theta,N}\}_{\theta\in\C,N\in\R_{\ge0}},&\nonumber\\
&\Theta:=\R_{\ge0},~\Xi_1:=\phi,~\Xi_3=G:=\C,~V:=W,~U:=U_n,&\nonumber
\end{eqnarray}
where $\phi$ is the empty set, $U_n$ is the concentrating operator satisfying $(\ref{concentrating})$, and $W$ is the mean shift operators $\{W_{\theta}\}_{\theta\in\C}$ as a representation of $\C$ defined in $(\ref{mean shift operator})$. Hence, a UMP min-max test for $(\mathrm{H}\textrm{-}6)$ is given as $U'^*_n(I\otimes T')U_{n}$ by using a UMP min-max test $T'$ with level $\alpha$ for 
\begin{eqnarray}
&H_0:r\in\{0\}~vs.~H_1:r\in(0,\infty)
~\hbox{with}~\{\rho_{0,N}^{\otimes (n-1)}\}_{N\in\R_{>0}}.& \nonumber
\end{eqnarray}
Since this hypothesis testing problem is nothing but $(\mathrm{H}\textrm{-}\mathrm{\chi^2})$, 
Theorem \ref{chi2 test} guarantees that 
$T_{\alpha}^{[\mathrm{\chi^2}],n-1}$ is a UMP min-max test for the above testing problem.
Therefore, we get the following proposition.

\begin{prp}\Label{num.min-max}

For the hypothesis $(\mathrm{H}\textrm{-}6)$, the test 
\[
T_{\alpha,N_0}^{[6],n}:=U_n^{\ast}(I\otimes T_{\alpha}^{[\mathrm{\chi^2}],n-1})U_n
\]
is a UMP min-max test with level $\alpha$.

\end{prp}

%%%%%%%%%%%%%%%%%%%%%%%%%%%%%%%%%%%%%%%%%%%%%%%%%%%%%%%%%%%%%%%
%%%%%%%%%%%%%%%%%%%%%%%%%%%%%%%%%%%%%%%%%%%%%%%%%%%%%%%%%%%%%%%
\section{Hypothesis testing of the number parameters for two kinds of quantum Gaussian states}

We consider the hypothesis testing about the consistency of the number parameters for two kinds of quantum Gaussian states.

%%%%%%%%%%%%%%%%%%%%%%%%%%%%%%%%%%%%%%%%%%%%%%%%%%%%%%%%%%%%%%%%%%%%%
%%%%%%%%%%%%%%%%%%%%%%%%%%%%%%%%%%%%%%%%%%%%%%%%%%%%%%%%%%%%%%%%%%%%%
\subsection{The case that the mean parameters are known: $F$ test}

~

In this subsection, we treat the hypothesis testing problem
\[
H_0:M=N~vs.~H_0:M\ne N~\hbox{with}~\{\rho_{\theta,M}^{\otimes m}\otimes \rho_{\eta,N}^{\otimes n}\}_{M,N\in\R_{> 0}} \eqno{(\mathrm{H}\textrm{-}7)}
\]
for $1\le n\in\N$ when the mean parameters $\theta,\eta$ is fixed. That is, we suppose that the mean parameters $\theta,\eta$ are known. Note that an optimal test for the classical analogue of $(\mathrm{H}\textrm{-}7)$ is given by an $F$ test.

We apply Lemma $\ref{lem1}$ to $(\mathrm{H}\textrm{-}7)$ in the following way:
% In Lemma $\ref{lem1}$, $(\mathrm{H}\textrm{-}7)$ is regarded as follows:
\[
\mathcal{S}:=\{\rho_{0,M}^{\otimes m}\otimes \rho_{0,N}^{\otimes n}\}_{M,N\in\R_{\ge0}},~\Theta:=\R_{\ge0}^{2},~\Xi_1:=\phi,~U:=W_{-\theta}^{\otimes m}\otimes W_{-\eta}^{\otimes n},
\]
where $\phi$ is the empty set and $W_{\theta}$ is the mean shift operator defined in $(\ref{mean shift operator})$. Hence, a UMP unbiased test for $(\mathrm{H}\textrm{-}7)$ is given as $(W_{-\theta}^{\otimes m}\otimes W_{-\eta}^{\otimes n})^*T'(W_{-\theta}^{\otimes m}\otimes W_{-\eta}^{\otimes n})$ by using a UMP unbiased test $T'$ with level $\alpha$ for 
\[
H_0:M=N~vs.~H_0:M\ne N~\hbox{with}~\{\rho_{0,M}^{\otimes m}\otimes \rho_{0,N}^{\otimes n}\}_{M,N\in\R_{> 0}}. \eqno{(\mathrm{H}\textrm{-}7)}
\]
Since this hypothesis testing problem is nothing but $(\mathrm{H}\textrm{-}\mathrm{F})$, 
Theorem \ref{F test} guarantees that 
$T_{\alpha}^{[\mathrm{F}],m,n}$ is a UMP unbiased test for the above testing problem.
Therefore, we get the following proposition.

\begin{prp}\Label{2numUMP}
For the hypothesis testing problem $(\mathrm{H}\textrm{-}7)$, the test 
\[
T_{\alpha,\theta,\eta}^{[7],m,n}:=(W_{-\theta}^{\otimes m}\otimes W_{-\eta}^{\otimes n})^{\ast} T_{\alpha}^{[\mathrm{F}],m,n} (W_{-\theta}^{\otimes m}\otimes W_{-\eta}^{\otimes n})
\]
 is a UMP unbiased test with level $\alpha$.
\end{prp}

%%%%%%%%%%%%%%%%%%%%%%%%%%%%%%%%%%%%%%%%%%%%%%%%%%%%%%%%%%%%%%%%%%%%%
%%%%%%%%%%%%%%%%%%%%%%%%%%%%%%%%%%%%%%%%%%%%%%%%%%%%%%%%%%%%%%%%%%%%%
\subsection{The case that the mean parameters are unknown: $F$ test}

~

In this subsection, we treat the hypothesis testing problem
\[
H_0:M=N~vs.~H_0:M\ne N~\hbox{with}~
\{\rho_{\theta,M}^{\otimes m}\otimes \rho_{\eta,N}^{\otimes n}\}_{\theta,\eta\in\C,M,N\in\R_{> 0}} \eqno{(\mathrm{H}\textrm{-}8)}
\]
for $2\le n\in\N$. 

We apply Theorem $\ref{thm3}$ to $(\mathrm{H}\textrm{-}8)$ in the following way:
%. In Theorem $\ref{thm3}$, $(\mathrm{H}\textrm{-}8)$ is regarded as follows:
\begin{eqnarray}
&\mathcal{S}_1:=\{\rho_{0,M}^{\otimes (m-1)}\otimes\rho_{0,N}^{\otimes (n-1)}\}_{M,N\in\R_{\ge 0}}, \mathcal{S}_2:=\phi, \mathcal{S}_3:=\{\rho_{\sqrt{m}\theta,M}\otimes\rho_{\sqrt{n}\eta,N}\}_{\theta,\eta\in\C,M,N\in\R_{> 0}},&\nonumber\\
&\Theta:=\R_{\ge0}^2,~\Xi_1:=\phi,~\Xi_3=G:=\C^2,~V:=W^{\otimes 2},~U:=U''_{m,n},~&\nonumber
\end{eqnarray}
where $\phi$ is the empty set, $U''_{m,n}$ is the unitary operator satisfying $(\ref{U3})$, and $W^{\otimes 2}$ is the mean shift operators $\{W_{\theta}\otimes W_{\eta}\}_{\theta,\eta\in\C}$ as a representation of $\C^2$ defined in $(\ref{mean shift operator})$. Hence, a UMP unbiased min-max test for $(\mathrm{H}\textrm{-}8)$ is given as $U''^{*}_{m,n}(I^{\otimes 2}\otimes T')U''_{m,n}$ by using a UMP unbiased min-max test $T'$ with level $\alpha$ for 
\[
H_0:M=N~vs.~H_0:M\ne N~\hbox{with}~
\{\rho_{0,M}^{\otimes (m-1)}\otimes \rho_{0,N}^{\otimes (n-1)}\}_{\theta,\eta\in\C,M,N\in\R_{> 0}}. 
\]
Since this hypothesis testing problem is nothing but $(\mathrm{H}\textrm{-}\mathrm{F})$, 
Theorem \ref{F test} guarantees that 
$T_{\alpha}^{[\mathrm{F}],m-1,n-1}$ is a UMP unbiased min-max test for the above testing problem.
Therefore, we get the following proposition.

\begin{prp}\Label{2numUMP2}
For the hypothesis testing problem $(\mathrm{H}\textrm{-}8)$, the test
\[
T_{\alpha}^{[8],m,n}={U''}_{m,n}(I^{\otimes 2}\otimes T_{\alpha}^{[\mathrm{F}],m-1,n-1}){U''}_{m,n}^{\ast}
\]
 is a UMP unbiased min-max test with level $\alpha$.
\end{prp}

%%%%%%%%%%%%%%%%%%%%%%%%%%%%%%%%%%%%%%%%%%%%%%%%%%%%%%%%%%%%%%%%
%%%%%%%%%%%%%%%%%%%%%%%%%%%%%%%%%%%%%%%%%%%%%%%%%%%%%%%%%%%%%%%%
\section{Comparison between the heterodyne measurement and the number measurement}

In this section, we compare the efficiency of the number measurement with that of the heterodyne measurement.

The heterodyne measurement $M_H=\{M_H(\xi)\}_{\xi\in\C}$ is defined by
\[
M_H(\xi):=\frac{1}{\pi}|\xi)(\xi|~~~(\xi\in\C)
\]
where $|\xi)(\xi|$ is a coherent state. This is a measurement widely used in an optical system, and a UMVUE for the mean parameter of quantum Gaussian states $\{\rho_{\theta,N}\}_{\theta\in\C}$ with the known number parameter $N$ \cite{Y-L}, that is, $M_H$ is optimal in the sense of quantum estimation.

When a system with a quantum Gaussian state $\rho_{\theta,N}$ is measured by the heterodyne measurement, its measured value is distributed according to two-dimensional Gaussian distribution with the mean parameter $(\re \theta, \im \theta)$ and the covariance matrix $\frac{N+1}{2}I_2$ where $I_2$ is the $2\times 2$ identity matrix.

%%%%%%%%%%%%%%%%%%%%%%%%%%%%%%%%%%%%%%%%%%%%%%%%%%%%%%%%%%%%%%%%
%%%%%%%%%%%%%%%%%%%%%%%%%%%%%%%%%%%%%%%%%%%%%%%%%%%%%%%%%%%%%%%%
\subsection{Comparison in hypothesis testing of the mean parameter}

~

We compare the type II error probabilities of tests based on the heterodyne measurement and the number measurement in the hypothesis testing problem
\[
H_0:|\theta|\in[0,R_0]~\text{vs.}~H_1:|\theta|\in(R_0,\infty)
\hbox{ with } 
\{\rho_{\theta,N}^{\otimes n}\}_{\theta\in\C}
 \eqno{(\mathrm{H}\textrm{-}1)}
\]
 when $R_0=0$, level $\alpha=0.1, n=1, N=\frac{1}{9}$.

Since the heterodyne measurement is a UMVUE, the decision based on the heterodyne measurement is expected to have good efficiency. 
The normal line of Fig \ref{f1} shows the type II error probability of the classical optimal test for the measured value obtained from the heterodyne measurement. 
On the other hand, Proposition \ref{mean.min-max} says that the optimal test in min-max criterion is based on the number measurement, 
and the thick line of Fig \ref{f1} shows the type II error probability of UMP min-max test in Proposition \ref{mean.min-max}.
This comparison shows that 
our optimal test much improves the combination of the optimal measurement for estimation and the classical optimal test
in the test $(\mathrm{H}\textrm{-}1)$.

\begin{figure}[htbp]
\begin{center}
\scalebox{1.0}{\includegraphics[scale=1]{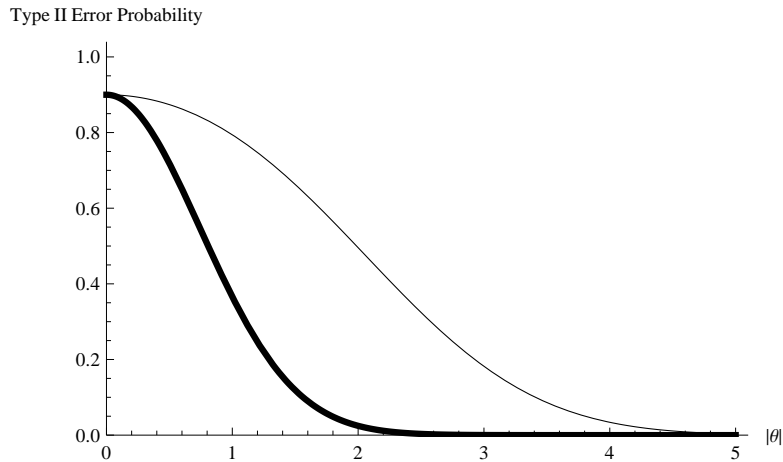}}
\end{center}
\caption{
The thick and normal lines show the type II error probabilities of the test based on the number measurement and the heterodyne measurement respectively when  $R_0=0$, $n=1, N=\frac{1}{9}$ and level $\alpha=0.1$ in Proposition $\ref{mean.min-max}$.}
\Label{f1}
\end{figure}%

~

%%%%%%%%%%%%%%%%%%%%%%%%%%%%%%%%%%%%%%%%%%%%%%%%%%%%%%%%%%%%%%%%
%%%%%%%%%%%%%%%%%%%%%%%%%%%%%%%%%%%%%%%%%%%%%%%%%%%%%%%%%%%%%%%%
\subsection{Comparison in hypothesis testing of the number parameter}

~

We compare the type II error probabilities of tests based on the heterodyne measurement and the number measurement in the hypothesis testing problem
\[
H_0:N\in[0,N_0]~vs.~H_1:N\in (N_0,\infty)~\hbox{with}~\{\rho_{\theta,N}^{\otimes n}\}_{N\in\R_{>0}} \eqno{(\mathrm{H}\textrm{-}5)}
\]
 when  $N_0=\frac{1}{9}, \alpha=0.1, n=1, \theta=0$.

Proposition \ref{numUMP} says that the optimal test is based on the number measurement.
Therefore, the test based on the heterodyne measurement have a larger type II error probability than the UMP test in Proposition \ref{numUMP}. 
The difference is showed in Fig \ref{f2}. 
The thick line of Fig \ref{f2} shows the type II error probability of UMP test in Proposition \ref{numUMP} and 
the normal line of Fig \ref{f2} shows the type II error probability of the classical optimal test for the measured value obtained from the heterodyne measurement. 
This comparison shows that 
our optimal test much improves the combination of the heterodyne measurement and the classical optimal test
in the test. $(\mathrm{H}\textrm{-}5)$.

Both comparisons indicate the importance of the number measurement in quantum hypothesis testing.

\begin{figure}[htbp]
\begin{center}
\scalebox{1.0}{\includegraphics[scale=1]{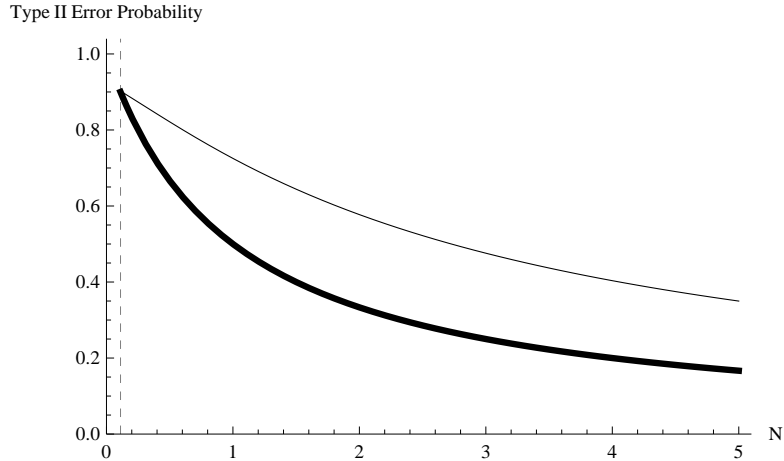}}
\end{center}
\caption{
The thick and normal lines show the type II error probabilities of the test based on the number measurement and the heterodyne measurement respectively when  $N_0=\frac{1}{9}$, $n=1, \theta=0$ and level $\alpha=0.1$ in Proposition $\ref{numUMP}$. The dashed line shows $N=N_0=\frac{1}{9}$.}
\Label{f2}
\end{figure}%

\section{Relation to quantum estimation for quantum Gaussian states}

~

Since the composite hypothesis testing treats a parametric state family with a disturbance parameter,
this problem is related to 
state estimation for a parametric state family with a disturbance parameter.
That is, 
we can expect that the measurement for an optimal test gives the optimal estimation based on this relation.
In order to treat this relation, we focus on a general family $\{\rho_{\theta,\eta}\}_{\theta\in\Theta,\xi\in \Xi}$ of quantum states on a quantum system $\mathcal{H}$,
in which, $\theta\in \Theta\subset\R$ is the parameter to be estimated and 
$\xi\in \Xi$ is the disturbance parameter. 
When a POVM $M$ takes values in $\R$ and satisfies 
\[
\int_{\R}x\mathrm{Tr}\rho_{\theta,\xi} M(dx)=\theta
\]
for all $\theta\in \Theta$ and $\xi\in \Xi$, it is called an {\it unbiased estimator}.

The error of an unbiased estimator $M$ is measured by the MSE (mean squared error) $V_{\theta,\xi}[M]:=\int_{\R^l}(x-\theta)^2\mathrm{Tr}\rho_{\theta,\xi}M(dx)$. 
The unbiased estimator $M$ is called \textbf{UMVUE$($Uniformly Minimum Variance Unbiased Estimator$)$}
when its MSE is smaller than the MSEs of other unbiased estimators,
i.e., 
any unbiased estimator $M'$ satisfies
\[
V_{\theta,\xi}[M]\le V_{\theta,\xi}[M']
\]
for all $\theta\in \Theta$ and $\xi\in \Xi$.

\begin{thm}\Label{UMVUE}
When the number parameter $N$ is to be estimated and the mean parameter $\theta$ is the disturbance parameter 
in the quantum Gaussian state family $\{\rho_{\theta,N}^{\otimes n}\}_{\theta\in\C,N\in\R_{> 0}}~(n\ge 2)$,
the POVM $M^{n-1}_{num}=\left\{M\left(\frac{k}{n-1}\right)\right\}_{k\in\Z_{\ge 0}}$ defined as follows is a UMVUE.
\[
M\left(\frac{k}{n-1}\right) :=\displaystyle\sum_{k_1+\cdot\cdot\cdot+k_{n-1}=k}U_n^{\ast}(I\otimes M_N(k_1)\otimes\cdot\cdot\cdot \otimes M_N(k_{n-1}))U_n.
\]

\end{thm}

\begin{proof}
Let $M=\{{M}(\omega)\}_{\omega\in\Omega}$ be an unbiased estimator for the number parameter $N$. 
Then
\[
X_M:=U_n^{\ast}\int_{\Omega}\omega M(d\omega)U_n
\]
satisfies 
\begin{align*}
&N
 =\int_{\Omega}\omega\mathrm{Tr}\rho_{\theta,N}^{\otimes n}M(d\omega)
 =\mathrm{Tr}\rho_{\theta,N}^{\otimes n}\int_{\Omega}\omega M(d\omega)\\
 =&\mathrm{Tr}\rho_{\sqrt{n}\theta,N}\otimes\rho_{0,N}^{\otimes (n-1)} X_M
 =\mathrm{Tr}\rho_{0,N}^{\otimes n}(W_{\sqrt{n}\theta}\otimes I^{\otimes (n-1)})^{\ast}X_M(W_{\sqrt{n}\theta}\otimes I^{\otimes (n-1)})
\end{align*}
for any $\theta \in \C$. Therefore, $X_M$ is represented as the form
\[
X_M=I\otimes Y_M.
\]
Then, the MSE of $M$ satisfies the following inequality.
\begin{align*}
 & (V_{\theta,N}[\mathrm{M}])_N
 =\int_{\Omega}(\omega-N)^2\mathrm{Tr}(\rho_{\theta,N}^{\otimes n}M(d\omega))\\
 =&\mathrm{Tr}(\rho_{\sqrt{n}\theta,N}\otimes\rho_{0,N}^{\otimes (n-1)}) U_n^{\ast}\int_{\Omega}\omega^2M(d\omega)U_n-N^2\\
 \ge & \mathrm{Tr}(\rho_{\sqrt{n}\theta,N}\otimes\rho_{0,N}^{\otimes (n-1)})X_M^2-N^2
 =\mathrm{Tr}\rho_{0,N}^{\otimes (n-1)}Y_M^2-N^2\\
 =&\mathrm{Tr}\rho_{0,N}^{\otimes (n-1)}(Y_M-N)^2
 \ge \Big(J^s_{N,n}\Big)^{-1},
\end{align*}
where $J^s_{N,n}$ is the SLD Fisher information metric for $\mathcal{S}':=\{\rho_{0,N}^{\otimes (n-1)}\}_{N\in\R_{> 0}}$. The last inequality is followed by the quantum Cram\'{e}r-Rao inequality (\cite{Hay5}). Since $\mathcal{S}'$ is the commutative family of quantum states, $\mathcal{S}'$ can be regarded as the family of probability distributions
\[
\mathcal{P}':=\left\{p_{N}(k_1,\cdot\cdot\cdot,k_{n-1}):=\left(\frac{1}{N+1}\right)\left(\frac{N}{N+1}\right)^{k_1+\cdot\cdot\cdot+k_{n-1}}\right\}_{N\in\R_{>0}}
\]
and $J^s_{N,n}$ coincides with the Fisher information metric $J_N$ on $\mathcal{P}'$. 
Therefore, $J^s_{N,n}$ can be calculated as follows.
\begin{align*}
& J^s_{N,n}
 =J_N
 =E_N[(\partial_N \mathrm{log}p_N)^2]
 =E_N[-\partial_N^2 \mathrm{log}p_N]\\
 =&E_N\Big[\left(\frac{1}{N}\right)^2\displaystyle\sum_{j=1}^{n-1}k_j-\left(\frac{1}{N+1}\right)^2\displaystyle\sum_{j=1}^{n-1}(k_j+1)\Big]
 =\frac{n-1}{N(N+1)}.
\end{align*}

The unbiasedness of $M^{n-1}_{num}$ is clear. The MSE $V_{\theta,N} $ of $M^{n-1}_{num}$ is
\begin{align*}
 &V_{\theta,N} 
=\mathrm{Tr}(\rho_{\sqrt{n}\theta,N}\otimes\rho_{0,N}^{\otimes (n-1)})
U_n^{\ast}\displaystyle\sum_{k\in\Z_{\ge 0}}\left(\frac{k}{n-1}\right)^2M\left(\frac{k}{n-1}\right)U_n-N^2\\
=&\mathrm{Tr}\rho_{0,N}^{\otimes (n-1)} \left(\displaystyle\sum_{k_1,\cdot\cdot\cdot,k_{n-1}\in\Z_{\ge
0}}\left(\frac{1}{n-1}\right)^2\left({\displaystyle\sum_{j=1}^{n-1}k_j}\right)^2M_N(k_1)\otimes\cdot\cdot\cdot\otimes
M_N(k_{n-1})\right)\\
&-N^2\\
=&\displaystyle\sum_{k_1,\cdot\cdot\cdot,k_{n-1}\in\Z_{\ge
0}}\left(\frac{1}{n-1}\right)^2\left({\displaystyle\sum_{j=1}^{n-1}k_j}\right)^2\left(\frac{1}{N+1}\right)^{n-1}(\frac{N}{N+1})^{k_1+\cdot\cdot\cdot+k_{n-1}}-N^2\\
=&\frac{1}{n-1}\displaystyle\sum_{k\in\Z_{\ge
0}}k^2\left(\frac{1}{N+1}\right)\left(\frac{N}{N+1}\right)^k
 +\frac{n-2}{n-1}\left(\displaystyle\sum_{k\in\Z_{\ge
0}}k\left(\frac{1}{N+1}\right)\left(\frac{N}{N+1}\right)^k\right)^2\\
 &-N^2\\
 =&(\frac{1}{n-1})^2\left((n-1)(2N^2+N)+(n-1)(n-2)N^2\right)-N^2 \\
 =& \frac{N(N+1)}{n-1}
 =(J^s_{N,n})^{-1}.
\end{align*}
The MSE of $M_{num}^{n-1}$ is the minimum in the unbiased estimators due to the quantum Cramer-Rao inequality, namely $M_{num}^{n-1}$ is the UMVUE.

\end{proof}

Theorem \ref{UMVUE} insists that
the following POVM is essential for inference of the number parameter for the $n$-copy family
$\rho_{0,N}^{\otimes n}$:
\[
M^{n}_{\chi^2}=\left\{M\left(\frac{k}{n}\right)=\displaystyle\sum_{k_1+\cdot\cdot\cdot+k_{n}=k}
M_N(k_1)\otimes\cdot\cdot\cdot \otimes M_N(k_{n})
\right\}_{k\in\Z_{\ge 0}}.
\]
We focus on the random variable 
$\bar{K}_{n,N}$ defined as the outcome of the POVM $M^{n+1}_{\chi^2}$ when the state is $\rho_{0,N}^{\otimes (n+1)}$. 
The random variable $n\bar{K}_{n, N}$ is distributed according to the negative binomial distribution:
\[
NB_{n,N}(k):=\binom{k+n-1}{n-1}\left(\frac{1}{N}\right)^n \left(\frac{N}{N+1}\right)^{k}~~~(k\in\Z_{\ge 0}).
\]
When a random variable $K_{n,N}$ is distributed according to $NB_{n,N}$, 
since the distribution of $\frac{2K_{n,N}}{N}$ plays the same role as the $\chi^2$ distribution in in estimation and testing of the number parameter $N$,
we call the distribution of $\frac{2K_{n,N}}{N}$ \textbf{$N$-$\chi^2$ distribution} with $n$ degrees of freedom and denote the distribution by $\chi^2_{n,N}$.

The distribution $\chi^2_{n,N}$ is a positive probability on the set $\frac{2}{N}\Z_{\ge 0}:=\{\frac{2}{N}l|l\in\Z_{\ge 0}\}$. 
Similar to the case of the $\chi^2$ distribution,
we call the maximum $r\in\frac{2}{N}\Z_{\ge 0}$ satisfying $P(K_{n,N}\ge r)\le \alpha$ the upper $\alpha$ point of $\chi^2_{n,N}$, which is denoted by $\chi^2_{n,N,\alpha}$. The threshold value $K_0$ used in (\ref{testfunction1}) is equal to $\frac{N_0}{2}\chi^2_{n,N,\alpha}$, and hence, the optimal tests for the hypothesis testing problem $(\mathrm{H}\textrm{-}5)$ and $(\mathrm{H}\textrm{-}6)$ are composed by $\chi^2_{n,N,\alpha}$. Note that $N$-$\chi^2$ distribution depends on the number parameter $N$ of the quantum Gaussian states 
although the (classical) $\chi^2$ distribution does not depend on the variance parameter of the Gaussian distributions. 

\begin{prp}

$\chi_{n, N}^2$ converges in distribution to $\chi_{2n}^2$ as $N\to\infty$ where $\chi_{2n}^2$ is the (classical) $\chi^2$ distribution with $2n$ degrees of freedom.

\end{prp}

\begin{proof}
In the first step, we prove the proposition when $n=1$. Let $F^{\chi_{2}^2}, F^{\chi_{2,N}^2}$ be the distributions of $\chi_{2}^2$, $\chi_{2, N}^2$ respectively. 
Then the following equation holds.
\[
F^{\chi_{2}^2}_N(x):=P\left(\frac{2K_{1, N}}{N}\le x\right)
 =\displaystyle\sum_{j=1}^{[\frac{Nx}{2}]}\frac{1}{N}\left(\frac{N}{N+1}\right)^{j-1}
 =1-\left(\frac{N}{N+1}\right)^{[\frac{Nx}{2}]}
\]
where [~] is the Gauss symbol. Hence
\[
\lim_{N\to\infty}F^{\chi_{2}^2}_N(x)=1-\lim_{N\to\infty}\frac{1}{{\left(1+\frac{x/2}{Nx/2}\right)^{\frac{Nx}{2}\frac{[Nx/2]}{Nx/2}}}}=1-e^{-\frac{x}{2}}=F^{\chi_{2}^2}(x).
\]

In the next step, we prove the proposition for an arbitrary $n\in\N$. Let $\frac{2k_{N, j}}{N}~(j=1, \cdot\cdot\cdot, n)$ be independent random variables distributed according to $\chi_{2, N}^2$. Then $\frac{2K_{n, N}}{N}$ is distributed according to $\chi_{n, N}^2$ since the distribution of $K_{n, N}:=\displaystyle\sum_{j=1}^n k_{N, j}$ is the negative binomial distribution. Every $\frac{2k_{N, j}}{N}$ can be taken to be independent and converge almost surely to $X_j$ distributed according to $\chi_{2}^2$ due to Skorokhod's representation theorem. Therefore $\chi_{n, N}^2$ converges in probability to $\chi_{2n}^2$ since $\frac{2K_{n, N}}{N}$ converges almost surely to $\displaystyle\sum_{j=1}^n X_j$.

\end{proof}

By the above proposition, $N$-$\chi^2$ distribution includes usual $\chi^2$ distribution as the limit. Since the number parameter $N$ in a quantum Gaussian state $\rho_{\theta,N}$ can be regarded as the number of photon in the photonic system, the limit $N\to\infty$ means the classical limit and a classical situation appears in the limit.

%%%%%%%%%%%%%%%%%%%%%%%%%%%%%%%%%%%%%%%%%%%%%%%%%%%%%%%%%%%%%%%%
%%%%%%%%%%%%%%%%%%%%%%%%%%%%%%%%%%%%%%%%%%%%%%%%%%%%%%%%%%%%%%%%
\section{Conclusion}
We have treated several composite quantum hypothesis testing problem with disturbance parameters in the quantum Gaussian system.
For this purpose, 
we have derived optimal tests in four fundamental quantum hypothesis testing problems
$(\mathrm{H}\textrm{-}\mathrm{A})$, $(\mathrm{H}\textrm{-}\mathrm{t})$, $(\mathrm{H}\textrm{-}\mathrm{\chi^2})$,
and $(\mathrm{H}\textrm{-}\mathrm{F})$, which are given as 
Theorems \ref{thmH-A}, \ref{t-test}, \ref{chi2 test}, and \ref{F test}, respectively.
We have also established a general theorem reducing complicated problems to fundamental problems (Theorem \ref{thm3}).
In the above both steps, group symmetry plays important roles, in which quantum Hunt-Stein Theorem is applied.
Combining both steps, we have derived optimal tests 
(UMP tests, UMP min-max tests, UMP unbiased tests, or UMP unbiased min-max tests)
for respective hypothesis testing problems on quantum Gaussian states. 
Since the quantum Gaussian state in quantum system corresponds to the Gaussian distribution in classical system, 
our testing problems for quantum Gaussian states play the same role as the testing problems for Gaussian distributions in the classical hypothesis testing. 

One may think that the above strategy is different from the processes deriving $t$, $\chi^2$ and $F$ tests for classical Gaussian distribution families.
A test is called $t$, $\chi^2$ and $F$ test in the classical case when its rejection region is determined by a statistics obeying $t$, $\chi^2$ and $F$ distribution.
In this case, the statistics can be calculated from plural statistics obeying the Gaussian distribution.
This calculation can be divided into the two parts:
The first part is an orthogonal linear transformation and ignoring information-less part, but depends on the problem.
The second part is applying non-linear transformations, but does not depend on the problem.
That is, the second part is common among all $t$, $\chi^2$ and $F$ tests, respectively.
The former of the quantum setting corresponds to the second part of the classical case, and 
the latter of the quantum setting corresponds to the first part of the classical case.
Hence, our quantum $t$, $\chi^2$ and $F$ tests can be regarded as suitable quantum versions of $t$, $\chi^2$ and $F$ tests.

However, simple application of the above two processes did not yield our optimal solution in the case of quantum $\chi^2$ and $F$ tests.
In these cases, after applying both processes, we have employed known facts for testing in classical exponential families.
This kind of application of classical result is the final important step in our derivation.

From the above characterization, we can find that it is important to remove 
several disturbance parameters and simplify our problem by the following two methods for quantum hypothesis testing.
The first method is applying quantum Hunt-Stein theorem for quantum hypothesis testing with group symmetry.
The second method is reducing complicated testing problems to fundamental testing problems, 
which contains application of quantum Hunt-Stein Theorem.
Indeed, 
a meaningful quantum hypothesis testing problem is not simple but composite, and have disturbance parameters. 
Since these two methods are very general, 
we can expect to apply our methods to such a complicated meaningful quantum hypothesis testing problem.

Our obtained optimal tests in the testing problems on quantum Gaussian states are implemented as measuring by the number measurement after performing the mean shift operator or the concentrating operator. 
In an optical system, the mean shift operation is approximately realized by using the beam splitter and the local oscillator (\cite{W-M}, p323). Similarly, the concentrating operation is realized by using the beam splitter. Therefore, the optimal tests given in this paper will be realizable since the number measurement can be prepared in an optical system.

We have focused on the symmetry in quantum hypothesis testing as one of the main interest of this paper
as well as existing studies with composite hypotheses\cite{Hay1,HMT}. 
But, in classical hypothesis testing, many optimal tests are derived without considering symmetry of the hypotheses. 
Therefore, in quantum hypothesis testing, 
it will be another challenging problem to present the existence conditions and 
the construction methods of obtained optimal tests without symmetry.

\section*{Acknowledgment}
The authors would like to thank Professor Fumio Hiai and Dr. Masaki Owari  for teaching a proof of the weak compactness theorem in Appendix and the realization of the mean shift operation in an optical system, respectively.
WK acknowledges support from Grant-in-Aid for JSPS Fellows No. 233283. MH is partially supported by a MEXT Grant-in-Aid for Young Scientists (A) No. 20686026 and Grant-in-Aid for Scientific Research (A) No. 23246071. 
The Center for Quantum Technologies is funded by the Singapore Ministry of Education and the National Research Foundation as part of the Research Centres of Excellence programme.

\appendix

\section{Appendix}

%%%%%%%%%%%%%%%%%%%%%%%%%%%%%%%%%%%%%%%%%%%%%%%%%%%%%%%%%%%%%%%%%%%%%
%%%%%%%%%%%%%%%%%%%%%%%%%%%%%%%%%%%%%%%%%%%%%%%%%%%%%%%%%%%%%%%%%%%%%
\subsection{Proof of quantum Hunt-Stein theorem}

\vspace{0.5em}

~

\hspace{-1.5em}\textit{{Proof of Theorem \ref{cpt.H-S}}:}
Since $\beta(\rho;T)=1-\mathrm{Tr}(\rho T)$ by the definition, we only has to show the following.
\[
\sup_{\tilde{T} \in {\cal T}_{\alpha,V}}\inf_{\xi \in \Xi} \mathrm{Tr}(\tilde{T}\rho_{\theta,\xi})
=\sup_{T\in {\cal T}_{\alpha}} \inf_{\xi \in \Xi}\mathrm{Tr}(T\rho_{\theta,\xi}) ,
\]
for any $\theta \in\Theta_1$ where ${\cal T}_{\alpha,V}$ is the set of tests of level $\alpha$
that are invariant concerning the (projective) representation $V$.

Let $T$ be a arbitrary test with level $\alpha$. There exists the invariant probability measure $\nu$ on $G$ since $G$ is compact. We denote the averaging test of $T$ with respect to $\nu$ by $\mathcal{L}(T):=\int_{g\in G}g\cdot T d\nu(g)$ where  $g\cdot X:=V_g X V_g^{\ast}$. Then $\tilde{T}$ is an invariant test concerning the representation $V$ with level $\alpha$ and holds the following inequality. For an arbitrary $\rho'$ satisfying $h(\rho)=h(\rho')$,
\[
\begin{array}{c}
\mathrm{Tr}(\rho_{\theta,\xi} \mathcal{L}(T))
 = \int_{g\in G}\mathrm{Tr}((g^{-1}\cdot\rho_{\theta,\xi}) {T})d\nu(g)
 \ge \int_{g'\in G}\displaystyle\inf_{g\in G}\mathrm{Tr}(\rho_{\theta,{g}^{-1}\cdot\xi} {T})d\nu(g')\\
 = \displaystyle\inf_{g\in G}\mathrm{Tr}(\rho_{\theta,{g}^{-1}\cdot\xi} {T})
 \ge \displaystyle\inf_{\xi\in\Xi}\mathrm{Tr}(\rho_{\theta,\xi} {T}).
\end{array}
\]
It implies 
\[
\displaystyle\inf_{\xi\in\Xi}\mathrm{Tr}(\rho_{\theta,\xi} \mathcal{L}(T)) 
\ge
\displaystyle\inf_{\xi\in\Xi}\mathrm{Tr}(\rho_{\theta,\xi} {T}).
\]
Since $\mathcal{L}(T)$ is an invariant test,
\[
\sup_{\tilde{T}\in {\cal T}_{\alpha,V}}\inf_{\xi\in\Xi}\mathrm{Tr}(\rho_{\theta,\xi} \tilde{T})
\ge \sup_{T \in {\cal T}_{\alpha}}\inf_{\xi\in\Xi}\mathrm{Tr}(\rho_{\theta,\xi} \mathcal{L}(T))
\ge \sup_{T \in {\cal T}_{\alpha}}\inf_{\xi\in\Xi}\mathrm{Tr}(\rho_{\theta,\xi} {T}).
\]

On the other hand, let $\tilde{T}$ be an arbitrary invariant test concerning the representation $f$ with level $\alpha$. Since 
\[
\inf_{\xi\in\Xi}\mathrm{Tr}(\rho_{\theta,\xi}\tilde{T})
\le \sup_{T}\inf_{\xi\in\Xi}\mathrm{Tr}(\rho_{\theta,\xi} {T}),
\]
clearly holds, we get
\[
\sup_{\tilde{T}\in {\cal T}_{\alpha,V}}\inf_{\xi\in\Xi}\mathrm{Tr}(\rho_{\theta,\xi}\tilde{T})
\le \sup_{T \in {\cal T}_{\alpha}}\inf_{\xi\in\Xi}\mathrm{Tr}(\rho_{\theta,\xi} {T}).
\]

\endproof
%\hspace{32.5em}■
\vspace{0.5em}

The key to prove Theorem \ref{cpt.H-S} is to be able to average a test $T$ by the invariant probability measure $\nu$ since $G$ is compact. But if $G$ is non-compact, there may not exist the invariant probability measure on $G$ and a test $T$ may not be able to be taken the average. Quantum Hunt-Stein theorem for a non-compact group insists that we can get the same consequence if we assume the existence of the asymptotic invariant probability measure $\{\nu_j\}_{j=1}^{\infty}$ on $G$ and a certain condition.

We prepare the following theorem to prove quantum Hunt-Stein theorem for a non-compact group.

\begin{thm}$(\textbf{Weak compactness theorem})$\Label{WCT}~
Let $\mathcal{H}$ be a Hilbert space of at most countable dimension. Then $\mathcal{T}:=\{T|0\le T\le I\}$ is weak $\ast$ sequentially compact. That is, for any sequence $\{T_n\}\subset\mathcal{T}$ there exists a partial sequence $\{T_{m_k}\}\subset\{T_n\}$ and an operator $T\in\mathcal{T}$ such that $\displaystyle\lim_{i\to\infty}\mathrm{Tr}XT_{n_j}=\mathrm{Tr} XT$ for any trace class operator $X$.

\end{thm}

Some propositions about topologies on spaces of operators are invoked in the following discussions. Let $\mathcal{H}$, $\langle,\rangle_{\mathcal{H}}$ and $\mathcal{B}(\mathcal{H})$ be a Hilbert space of at most countable dimension, an inner product on $\mathcal{H}$, and the space of bounded operators on $\mathcal{H}$ respectively.

A sequence $\{A_n\}_{n=1}^{\infty}\subset\mathcal{B}(\mathcal{H})$ is said to converge weakly to $A\in\mathcal{B}(\mathcal{H})$ if $\langle A_n x,y\rangle_{\mathcal{H}}$ converges to $\langle Ax,y\rangle_{\mathcal{H}}$ for any $x,y\in\mathcal{H}$. The topology defined by the above convergence is called weak operator topology. A sequence $\{A_n\}_{n=1}^{\infty}\subset\mathcal{B}(\mathcal{H})$ is said to converge to $A\in\mathcal{B}(\mathcal{H})$ in the weak $\ast$ topology if $\mathrm{Tr}BA_n$ converges to $\mathrm{Tr}BA$ for any $B\in\mathcal{C}_1(\mathcal{H}))$. The topology defined by the above convergence is called weak $\ast$ topology. In general weak $\ast$ topology is defined on the dual Banach space of a Banach space. It is clear that operator weak topology is equivalent to weak $\ast$ topology in a bounded closed set in $\mathcal{B}(\mathcal{H})$.

\vspace{0.5em}
\hspace{-1.5em}\textit{{Proof of Theorem \ref{WCT}}:}
The closed unit ball $D$ in $\mathcal{B}(\mathcal{H})=\mathcal{C}_1(\mathcal{H})^{\ast}$ is compact with respect to weak $\ast$ topology due to Banach-Alaoglu theorem. Because $\mathcal{T}$ is closed subset of $D$, $\mathcal{T}$ is compact in weak $\ast$ topology. Since $\mathcal{H}$ have at most countable dimension, a bounded closed set $\mathcal{T}$ in $\mathcal{B}(\mathcal{H})$ is metrizable with respect to weak operator topology in $\mathcal{B}(\mathcal{H})$. Therefore $\mathcal{T}$ is sequentially compact in weak $\ast$ topology because weak $\ast$ topology is equivalent to operator weak topology in a bounded closed subset $\mathcal{T}$ of $\mathcal{B}(\mathcal{H})$. 

\endproof
%\hspace{32.5em}■
\vspace{0.5em}

Quantum Hunt-Stein theorem is proven for a complete family of quantum states with non-compact action by using the weak compactness theorem.

\vspace{0.5em}
\hspace{-1.5em}\textit{{Proof of Theorem \ref{H-S}}:}
Set $T_n:=\int_G (g\cdot T)d\nu_n(g)\in\mathcal{T}_{\alpha}$ for $T \in \mathcal{T}_{\alpha}$ where $g\cdot X:=V_g X V_g^{\ast}$. Due to the weak compactness theorem, there exists a sequence of $\{T_n\}$ and a $\tilde{T}\in\mathcal{T}_{\alpha}$ such that $T_{n_j}$ converges to $\mathcal{L}(T)$ in weak $\ast$ topology. Fix $m\in\N$ and set
\[
B_k(\rho_{\theta,\xi})=\left\{h\in G\Big|\frac{k-1}{m}\le \mathrm{Tr}(\rho_{\theta,\xi} (h\cdot{T}))\le \frac{k}{m}\right\}
\]
for a state $\rho_{\theta,\xi}\in\mathcal{S}$ and $k\in\{0,1,\cdot\cdot\cdot,m\}$.
\[
\begin{array}{c}
\displaystyle\sum_{k=0}^m\frac{k-1}{m}\nu_{n_j}(B_k(\rho_{\theta,\xi}))
 \le 
\displaystyle\sum_{k=0}^m\int_{B_k(\rho_{\theta,\xi})}\mathrm{Tr}(\rho_{\theta,\xi} (h\cdot{T}))d\nu_{n_j}(h)\\
 =
\mathrm{Tr}(\rho_{\theta,\xi} \mathcal{L}(T))
 \le
\displaystyle\sum_{k=0}^m\frac{k}{m}\nu_{n_j}(B_k(\rho_{\theta,\xi}))
 =
\displaystyle\sum_{k=0}^m\frac{k-1}{m}\nu_{n_j}(B_k(\rho_{\theta,\xi}))+\frac{1}{m}
\end{array}
\]
implies
\[
|\mathrm{Tr}(\rho_{\theta,\xi} \mathcal{L}(T))-\displaystyle\sum_{k=0}^m\frac{k-1}{m}\nu_{n_j}(B_k(\rho_{\theta,\xi}))|\le \frac{1}{m}.
\]
Similarly,
\[
|\mathrm{Tr}(\rho_{\theta,\xi}
(g\cdot\mathcal{L}(T)))-\displaystyle\sum_{k=0}^m\frac{k-1}{m}\nu_{n_j}(B_k(\rho_{\theta,\xi})g^{-1})|\le \frac{1}{m}
\]
is derived. Therefore
\[
\begin{array}{ll}
&|\mathrm{Tr}(\rho_{\theta,\xi} (g\cdot\mathcal{L}(T)))-\mathrm{Tr}(\rho_{\theta,\xi} \mathcal{L}(T))|\\
\le&|\mathrm{Tr}(\rho_{\theta,\xi}
(g\cdot\mathcal{L}(T)))-\displaystyle\sum_{k=0}^m\frac{k-1}{m}\nu_{n_j}(B_k(\rho_{\theta,\xi})g^{-1})|\\
&+|\displaystyle\sum_{k=0}^m\frac{k-1}{m}\nu_{n_j}(B_k(\rho_{\theta,\xi})g^{-1})-\displaystyle\sum_{k=0}^m\frac{k-1}{m}\nu_{n_j}(B_k(\rho_{\theta,\xi}))|\\
 &+|\displaystyle\sum_{k=0}^m\frac{k-1}{m}\nu_{n_j}(B_k(\rho_{\theta,\xi}))-\mathrm{Tr}(\rho_{\theta,\xi} \mathcal{L}(T))|\\
 \le&\frac{2}{m}+\displaystyle\sum_{k=0}^m\frac{k-1}{m}|\nu_{n_j}(B_k(\rho_{\theta,\xi})g^{-1})-\nu_{n_j}(B_k(\rho_{\theta,\xi}))|.
\end{array}
\]
We get
\[
|\mathrm{Tr}(\rho_{\theta,\xi} (g\cdot\mathcal{L}(T)))-\mathrm{Tr}(\rho_{\theta,\xi} \mathcal{L}(T))|\le\frac{2}{m}
\]
as $j\to\infty$. By the arbitrariness of $m\in\N$,
\[
\mathrm{Tr}(\rho_{\theta,\xi} (g\cdot\mathcal{L}(T)))=\mathrm{Tr}(\rho_{\theta,\xi} \mathcal{L}(T))
\]
holds. It follows that
\[
g\cdot\mathcal{L}(T)=\mathcal{L}(T)
\]
since the completeness of the family of quantum states $\mathcal{S}$. The above equation implies that $\mathcal{L}(T)$ is $G$-invariant.

Because $T_{n_j}$ converges to $\mathcal{L}(T)$ in the weak $\ast$ topology, we get
\[
\displaystyle\lim_{j\to\infty}\mathrm{Tr}(\rho_{\theta,\xi} {T_{n_j}})=\mathrm{Tr}(\rho_{\theta,\xi} \mathcal{L}(T))~~~(\rho_{\theta,\xi}\in\mathcal{S}).
\]
Hence
\[
\begin{array}{c}
\mathrm{Tr}(\rho_{\theta,\xi} {T_{n_j}})
 =\int_G \mathrm{Tr}(\rho_{\theta,\xi} (g\cdot T))d\nu_{n_j}(g)
 =\int_G \mathrm{Tr}(\rho_{\theta,g^{-1}\cdot\xi}T)d\nu_{n_j}(g),
\end{array}
\]
and the above equation implies
\[
\displaystyle\inf_{g\in G}\mathrm{Tr}(\rho_{\theta,g^{-1}\cdot\xi} {T})\le \mathrm{Tr}(\rho_{\theta,\xi} {T_{n_j}}).
\]
As $j\to\infty$, we get
\[
\inf_{\xi\in\Xi}\mathrm{Tr}(\rho_{\theta,\xi} {T})\le
\displaystyle\inf_{g\in G}\mathrm{Tr}(\rho_{\theta,g^{-1}\cdot\xi} {T})\le \mathrm{Tr}(\rho_{\theta,\xi} \mathcal{L}(T)).
\]
It implies
\[
\inf_{\xi\in\Xi}\mathrm{Tr}(\rho_{\theta,\xi} \mathcal{L}(T)) \ge \inf_{\xi\in\Xi}\mathrm{Tr}(\rho_{\theta,\xi} {T}).
\]
Therefore
\[
\sup_{\tilde{T}\in\mathcal{T}_{\alpha,V}}\inf_{\xi\in\Xi}\mathrm{Tr}(\rho_{\theta,\xi} {\tilde{T}})
\ge \sup_{T\in\mathcal{T}_{\alpha}}\inf_{\xi\in\Xi}\mathrm{Tr}(\rho_{\theta,\xi} {\mathcal{L}(T)})
\ge \sup_{T\in\mathcal{T}_{\alpha}}\inf_{\xi\in\Xi}\mathrm{Tr}(\rho_{\theta,\xi} {T}).
\]

On the other hand, since
\[
\inf_{\xi\in\Xi}\mathrm{Tr}(\rho_{\theta,\xi} {\tilde{T}})\le\sup_{T\in\mathcal{T}_{\alpha}}\inf_{\xi\in\Xi}\mathrm{Tr}(\rho_{\theta,\xi} {T})
\]
clearly holds, we get
\[
\displaystyle\sup_{\tilde{T}\in\mathcal{T}_{\alpha,V}}\inf_{\xi\in\Xi}\mathrm{Tr}(\rho_{\theta,\xi} {\tilde{T}})\le\sup_{T\in\mathcal{T}_{\alpha}}\inf_{\xi\in\Xi}\mathrm{Tr}(\rho_{\theta,\xi} {T}).
\] 
\endproof
%\hspace{32.5em}■
\vspace{0.5em}

%%%%%%%%%%%%%%%%%%%%%%%%%%%%%%%%%%%%%%%%%%%%%%%%%%%%%%%%%%%%%%%%%%%%%
%%%%%%%%%%%%%%%%%%%%%%%%%%%%%%%%%%%%%%%%%%%%%%%%%%%%%%%%%%%%%%%%%%%%%
\subsection{The nonexistence of a UMP test }

~

\begin{prp}

For the hypothesis testing problem $(\mathrm{H}\textrm{-}1)$, there exist $N\in\R_{> 0}$, ${R_0}\in\R_{>0}$ and level $\alpha\in [0,1]$ such that there is no UMP test.

\end{prp}

\begin{proof}
Assume that there exists a UMP test $T_{UMP}$ with level $\alpha$ for the hypothesis testing problem $H_0:r\le {R_0}~vs.~H_1:r> {R_0}$. Then $T_{\alpha,R_0}^{[1],1}$ in Theorem $\ref{mean.min-max}$ is proven to be a UMP test as follows. Set a $S^1$-invariant test concerning the representation ${V}^{[1],1}$ as
\[
\tilde{T}:=\frac{1}{2\pi}\int_0^{2\pi}e^{i\eta\hat{N}}T_{UMP} e^{-i\eta\hat{N}}d\eta.
\]

For $r\mathrm{e}^{it}\in\C$ satisfying $r>{R_0}$,
\[
\begin{array}{c}
\mathrm{Tr}(\rho_{r\mathrm{e}^{it}, N}\tilde{T})
 =\frac{1}{2\pi}\int_0^{2\pi}\mathrm{Tr}(\rho_{e^{-i\eta}r\mathrm{e}^{it}, N}T_{UMP})d\eta
 \ge\displaystyle\min_{\eta\in[0, 2\pi)}\mathrm{Tr}(\rho_{e^{-i\eta}r\mathrm{e}^{it}, N}T_{UMP})\\
 =\mathrm{Tr}(\rho_{e^{-i\eta_0}r\mathrm{e}^{it}, N}T_{UMP})
 \ge\mathrm{Tr}(\rho_{e^{-i\eta_0}r\mathrm{e}^{it}, N}\tilde{T})
 =\mathrm{Tr}(\rho_{r\mathrm{e}^{it}, N}\tilde{T})\\
\end{array}
\]
holds, that is, the above inequality is actually the equality. In particular, since
\[
\frac{1}{2\pi}\int_0^{2\pi}\mathrm{Tr}(\rho_{e^{-i\eta}r\mathrm{e}^{it}, N}T_{UMP})d\eta
 =\displaystyle\min_{\eta\in[0, 2\pi)}\mathrm{Tr}(\rho_{e^{-i\eta}r\mathrm{e}^{it}, N}T_{UMP}),
\]
we get
\[
\mathrm{Tr}(\rho_{e^{i\eta}r\mathrm{e}^{it}, N}T_{UMP})=\mathrm{Tr}(\rho_{r\mathrm{e}^{it}, N}T_{UMP})
\]
for any $\eta\in[0, 2\pi)$. Therefore
\[
\mathrm{Tr}(\rho_{r\mathrm{e}^{it}, N}\tilde{T})=\mathrm{Tr}(\rho_{r\mathrm{e}^{it}, N}T_{UMP}).
\]
Hence the $S^1$-invariant test $\tilde{T}$ concerning the representation ${V}^{[1],1}$ is a UMP test, and the test $T_{\alpha,R_0}^{[1],1}$ in Theorem $\ref{mean.min-max}$ is also a UMP test.

We set as
\[
T_1=\displaystyle\sum_{k=1}^{\infty}|k\rangle\langle k|, 
 ~~~T_2=\frac{1}{4}(|0\rangle+|1\rangle)(\langle 0|+\langle 1|)+\displaystyle\sum_{k=2}^{\infty}|k\rangle\langle k|.
\]
 Let $0<x_1<x_2$ be the solutions of the quadratic equation $x^2-\frac{32}{3}x+\frac{2}{3}=0$, and ${R_0}=x_1, \alpha=Tr(\rho_{x_1\mathrm{e}^{it}, 1}T_1)$. Then for the hypothesis testing problem 
\[
H_0:r\in [0,x_1] ~vs.~ H_1:r\in (x_1,\infty) ~\hbox{with}~ \{\rho_{r\mathrm{e}^{it}, 1}\}_{r\mathrm{e}^{it}\in\C},
\]
 $T_1$ is the UMP $S^1$-invariant test concerning the representation ${V}^{[1],1}$ with level $\alpha$. Since 
\[
\begin{array}{lll}
\mathrm{Tr}(\rho_{1, r}T_1)-\mathrm{Tr}(\rho_{1, r}T_2)
&=&\frac{3}{32}e^{-\frac{r^2}{2}}(r^2-\frac{32}{3}r+\frac{2}{3})
\end{array}
\]
for $r\in\R_{> 0}$, $T_2$ is a test with level $\alpha$ and has smaller type II error than $T_1$ in $r\in(x_1, x_2)$ because 
\[
\mathrm{Tr}(\rho_{r\mathrm{e}^{it},1}T_1)<\mathrm{Tr}(\rho_{r\mathrm{e}^{it},1}T_2).
\]
Since a UMP $S^1$-invariant test concerning the representation ${V}^{[1],1}$ is also a UMP test, the above inequality is a contradiction. It means that there exists no UMP test for $(\mathrm{H}\textrm{-}1)$ when $N=1, {R_0}=x_1$ and $\alpha=\mathrm{Tr}(\rho_{R_0\mathrm{e}^{it}, 1}T_1)$.

\end{proof}

\end{document}